%% file: main.tex
\documentclass[11pt]{article}
\usepackage[dvipsnames]{xcolor} 
\usepackage{graphicx} 
\usepackage{setspace}
\usepackage{geometry}
\usepackage{amsmath}
\usepackage{amssymb}
\usepackage{amsthm}
\usepackage{thmtools}
\usepackage{framed}
\usepackage{complexity}
\usepackage{tcolorbox}
\usepackage{enumitem}
\usepackage[hyperfootnotes=false,hypertexnames=false,colorlinks=true,linkcolor=Maroon,citecolor=Maroon,urlcolor=Maroon]{hyperref}
\usepackage{orcidlink}   

\usepackage[capitalize]{cleveref}   

\title{Weak Zero-Knowledge and One-Way Functions}
\author{Rohit Chatterjee\footnote{{\orcidlink{0009-0004-1970-1779}} {\tt rochat@nus.edu.sg}. Department of Computer Science, National University of Singapore.} \and Yunqi Li\footnote{{\orcidlink{0009-0001-3087-3190}} {\tt yunqili@comp.nus.edu.sg}. Department of Computer Science, National University of Singapore.} \and Prashant Nalini Vasudevan\footnote{\orcidlink{0000-0001-6880-795X} {\tt prashvas@nus.edu.sg}. Department of Computer Science, National University of Singapore.}}
\date{\today}

\newcommand{\ec}{\epsilon_c}
\newcommand{\es}{\epsilon_s}
\newcommand{\ez}{\epsilon_z}
\newcommand{\ecn}{\epsilon_c(n)}
\newcommand{\esn}{\epsilon_s(n)}
\newcommand{\ezn}{\epsilon_z(n)}

\newcommand{\cB}{\mathcal{B}}
\newcommand{\cC}{\mathcal{C}}
\newcommand{\cD}{\mathcal{D}}
\newcommand{\cF}{\mathcal{F}}
\renewcommand{\cL}{\mathcal{L}}
\newcommand{\cO}{\mathcal{O}}
\newcommand{\cT}{\mathcal{T}}
\newcommand{\cR}{\mathcal{R}}
\newcommand{\cU}{\mathcal{U}}
\renewcommand{\cP}{\mathcal{P}}

\newcommand{\ioP}{\mathsf{ioP/poly}}
\newcommand{\oabpp}{\mathsf{1AvgBPP/poly}}
\newcommand{\iooabpp}{\mathsf{io1AvgBPP/poly}}

\newcommand{\gen}{\mathsf{Gen}}
\newcommand{\sfP}{\mathsf{P}}
\newcommand{\sfV}{\mathsf{V}}

\newcommand{\Sim}{\mathsf{Sim}}

\newcommand{\tsfP}{\Tilde{\sfP}}

\newcommand{\tr}{\mathsf{View}\langle\mathsf{P}, \mathsf{V}\rangle}
\newcommand{\prot}[2]{\langle #1, #2 \rangle}
\newcommand{\Est}{\mathsf{Est}}
\newcommand{\tran}[1]{\mathsf{View}_{#1}\langle\sfP, \sfV \rangle}

\newcommand{\est}{\mathsf{est}}
\newcommand{\rand}{\mathsf{rd}}
\input{preamble}

\begin{document}

\pagenumbering{roman}
\thispagestyle{empty}

\maketitle

\input{abstract}

\newpage
\setcounter{tocdepth}{3}
\tableofcontents
\newpage

\pagenumbering{arabic}

\input{intro}

\input{pre}

\input{aiowfbig}

\input{owfbig}

\subsection*{Acknowledgements}

This work was supported by the National Research Foundation, Singapore, under its NRF Fellowship programme, award no. NRF-NRFF14-2022-0010.  

Commercial AI tools (ChatGPT, Claude) were used as typing assistants for grammar and basic editing, and for mild assistance with LaTeX formatting.

\bibliographystyle{alpha}
\bibliography{refs}

\appendix

\input{shortappendix}

\end{document}

%% file: preamble.tex
\newif\ifdraft
\draftfalse

\ifdraft
\newcommand{\ynote}[1]{[{\footnotesize \color{Green}{\bf Yunqi:} { {#1}}}]}
\newcommand{\pnote}[1]{[{\footnotesize \color{orange}{\bf Prashant:} { {#1}}}]}
\newcommand{\rnote}[1]{[{\footnotesize \color{magenta}{\bf Rohit:} { {#1}}}]}

\newcommand{\todo}[1]{\colorbox{yellow}{{\rm TODO #1}}}
\else
\newcommand{\pnote}[1]{}
\newcommand{\ynote}[1]{}
\newcommand{\rnote}[1]{}

\newcommand{\todo}[1]{}
\fi

\theoremstyle{definition}
\newtheorem{definition}    {Definition} [section]

\theoremstyle{plain}
\newtheorem{theorem}       {Theorem}    [section]
\newtheorem{claim}         {Claim}      [section]
\newtheorem{remark}        {Remark}     [section]

\newtheorem{lemma}         [theorem]    {Lemma}

\newtheorem{corollary}     [theorem]    {Corollary}

\newenvironment{proofof}[1]{\begin{proof}[\textit{Proof of #1}]}{\end{proof}}
\newenvironment{proof-sketch}{\begin{proof}[\textit{Proof Sketch}]}{\end{proof}}
\newenvironment{proof-sketch-of}[1]{\begin{proof}[\textit{Proof Sketch of #1}]}{\end{proof}}

\crefname{fact}{Fact}{Facts}
\crefname{claim}{Claim}{Claims}
\crefname{corollary}{Corollary}{Corollaries}

\newcommand{\set}[1]   {\left\{#1\right\}}
\newcommand{\abs}[1]   {\left|#1\right|}
\newcommand{\size}[1]  {\left|#1\right|}

\newcommand{\ra}       {\rightarrow}
\newcommand{\zo}       {\{0,1\}}

\newcommand{\Nat}      {\mathbb{N}}

\newcommand{\eps}      {\varepsilon}

\newcommand{\Prob}[2]  {\mathop{\mathrm{Pr}}_{#1}\left[#2\right]}

\newcommand{\Expec}[2] {\mathop{\mathbb{E}}_{#1}\left[#2\right]}

\newcommand{\negl}     {\mathsf{negl}}

\newcommand{\cA}{\mathcal{A}}

%% file: abstract.tex
\begin{abstract}
    We study the implications of the existence of weak Zero-Knowledge (ZK) protocols for worst-case hard languages. These are protocols that have completeness, soundness, and zero-knowledge errors (denoted $\epsilon_c$, $\epsilon_s$, and $\epsilon_z$, respectively) that might not be negligible. Under the assumption that there are worst-case hard languages in $\mathsf{NP}$, we show the following:
    \begin{enumerate}[itemsep=5pt,topsep=5pt]
        \item If all languages in $\mathsf{NP}$ have NIZK proofs or arguments satisfying $\epsilon_c+\epsilon_s+\epsilon_z < 1$, then One-Way Functions (OWFs) exist.  

        \vspace{5pt}
        This covers all possible non-trivial values for these error rates. It additionally implies that if all languages in $\mathsf{NP}$ have such NIZK proofs and $\epsilon_c$ is negligible, then they also have NIZK proofs where all errors are negligible. Previously, these results were known under the more restrictive condition $\epsilon_c+\sqrt{\epsilon_s}+\epsilon_z < 1$ [Chakraborty et al., CRYPTO 2025]. 
        
        \item If all languages in $\mathsf{NP}$ have $k$-round public-coin ZK proofs or arguments satisfying $\epsilon_c+\epsilon_s+(2k-1)\cdot\epsilon_z<1$, then OWFs exist.
        
        \item If, for some constant $k$, all languages in $\mathsf{NP}$ have $k$-round public-coin ZK proofs or arguments satisfying $\epsilon_c+\epsilon_s+k\cdot\epsilon_z<1$, then infinitely-often OWFs exist.
    \end{enumerate}
\end{abstract}

%% file: intro.tex
\section{Introduction}

The notion of Zero-Knowledge (ZK) protocols is a vital part of modern cryptography. These are interactive proof systems in which a \emph{prover} proves to a \emph{verifier} the validity of a given statement, with the additional guarantee that even a possibly cheating verifier cannot obtain any secrets that might have been used in the proof. More precisely, such protocols guarantee {\em soundness}, i.e. cheating provers cannot prove false statements, and {\em zero knowledge}, i.e. a cheating verifier cannot glean anything beyond the validity of the statement over the course of the protocol. 


\paragraph{Hardness from Zero-Knowledge} One natural question arising in the context of relating zero-knowledge to other cryptographic notions is that of which other cryptographic primitives are implied by it. This was first studied in the work of Ostrovsky~\cite{DBLP:conf/coco/Ostrovsky91}, who showed that a statistical ZK proof system for any average-case hard language implies the existence of One-Way Functions (OWFs). This was later generalized by Ostrovsky and Wigderson~\cite{DBLP:conf/istcs/OstrovskyW93} to computational ZK proofs, and they showed in addition that such proofs for even a worst-case hard language implies a weaker form of OWFs called auxiliary-input OWFs. 

Building on these, the recent work of Hirahara and Nanashima~\cite{DBLP:conf/stoc/HiraharaN24} showed that if all languages in $\NP$ have computational ZK proofs (or even arguments, which only have computational soundness) and $\NP$ is hard in the worst-case, then OWFs exist.

\paragraph{Weak Zero-Knowledge} The above results all work with proof systems where the completeness, soundness, and zero-knowledge errors are guaranteed to be negligible. There are, however, a number of natural and useful ZK protocols, such as the commonly taught protocols for 3-Coloring and Graph Non-isomorphism, that do not natively have negligible error rates. Such {\em Weak ZK} protocols may have completeness error $\epsilon_c$, soundness error of $\epsilon_{s}$, and zero-knowledge error $\epsilon_{z}$, that may be as large as a constant number. Along the same lines as above, it is natural and important to study the power of such protocols as well. 

\paragraph{Amplifying Weak NIZKs}
The work of Goyal et al.~\cite{DBLP:conf/crypto/GoyalJS19} was a first step in this direction, investigating the power of weak {\em Non-Interactive ZK} (NIZK) arguments. NIZKs consider a non-interactive setting where the prover and verifier have access to a common random string, and the protocol only involves a single prover message, following which the verifier decides whether to accept or reject. They showed that weak NIZKs with negligible $\eps_c$ satisfying $\epsilon_{s} + \epsilon_{z} < 1-\delta$ for any non-zero constant $\delta$ can be amplified to a standard NIZK argument with negligible errors if we additionally have access to a sub-exponentially secure Public Key Encryption (PKE) scheme. 

Bitansky and Geier~\cite{DBLP:conf/crypto/BitanskyG24} improved this result to show that standard PKE suffices for this amplification. They also showed that if the NIZK system is a proof (i.e., has statistical soundness), then amplification is possible assuming only OWFs. Applebaum and Kachlon~\cite{DBLP:conf/crypto/ApplebaumK25} improved on this to allow for $\delta$ to be as small as an inverse-polynomial function.

\paragraph{Hardness from Weak ZK} Seeking to reduce the assumptions needed for such amplification, Chakraborty et al.~\cite{DBLP:conf/crypto/ChakrabortyHK25} showed that in certain settings, weak NIZKs can be used to derive OWFs. Specifically, they show that if all languages in $\NP$ have NIZK arguments satisfying $\epsilon_{c}  + \sqrt{\epsilon_{s}}+ \epsilon_{z} < 1$, then OWFs exist under just the worst-case assumption that $\NP \not \subseteq \ioP$\footnote{That is, there are no polynomial-size circuit families for $\NP$, even if they are only required to be correct infinitely often}. They then combined this with the amplification results of \cite{DBLP:conf/crypto/BitanskyG24,DBLP:conf/crypto/ApplebaumK25} to show that under the additional hypotheses that these are NIZK proofs with negligible $\eps_c$, they could get NIZK proofs with negligible errors for all of $\NP$. 

While this result helps characterize the hardness of a broad class of weak NIZKs, it is still somewhat unsatisfactory as it does not cover {\em all} possible non-trivial weak NIZK parameters. As noted in \cite{DBLP:conf/crypto/ChakrabortyHK25}, the setting $\epsilon_{c} + \epsilon_{s} + \epsilon_{z} \ge 1$ is not meaningful for NIZK protocols. Thus, to complete the picture, what is required is to understand the implications of any weak NIZK protocol satisfying $\epsilon_{c} + \epsilon_{s} + \epsilon_{z} <1$.

Another important question that has not been studied in this recent line of work is that of the complexity of weak interactive ZK protocols. The earlier implications of ZK shown in \cite{DBLP:conf/istcs/OstrovskyW93,DBLP:conf/stoc/HiraharaN24} work for interactive ZK protocols with negligible errors. But the extent of their validity for weak ZK protocols was not understood.

\medskip
\paragraph{Our Results}

Our first result addresses the first question above, extending the construcion of OWFs from NIZKs to the most general parameters, under the same assumptions as in prior work.

\begin{theorem}[Informally, \cref{the:nizk_owf}]
    If $\NP \not \subseteq \ioP$, and every language in $\NP$ has an $(\ec, \es, \ez)$-NIZK proof (or argument) with $\ec+\es+\ez < 1$, then one-way functions exist. 
\end{theorem}

Similar to \cite{DBLP:conf/crypto/ChakrabortyHK25}, we can in turn use the amplification results of \cite{DBLP:conf/crypto/BitanskyG24,DBLP:conf/crypto/ApplebaumK25} to obtain amplification of NIZK proofs with near-perfect completeness, from the most general setting of the remaining errors. 

\begin{corollary}[NIZK Amplification]
    If $\NP \not \subseteq \ioP$, and every language in $\NP$ has an $(\ec, \es, \ez)$-NIZK proof with $\ec+\es+\ez < 1$ and negligible $\ec$, then every language in $\NP$ has a NIZK proof with negligible errors. Here, the soundness of the NIZK proofs is required to be adaptive.\footnote{In the rest of the paper, we exclusively use the weaker non-adaptive definition of soundness for NIZK protocols. This only strengthens our results, as we use NIZKs to construct other things. Here, the stronger adaptive notion of soundness is needed. See, e.g., \cite{DBLP:conf/crypto/BitanskyG24} for the definition of this notion.}
\end{corollary}

We also address the second question raised above for {\em public-coin} protocols, where the verifier's messages consist solely of uniform random bits.

\begin{theorem}[Informally, \cref{the:pczk_owf}]
    If $\NP \not \subseteq \ioP$, and every language in $\NP$ has a $t$-message $(\ec, \es, \ez)$-public-coin ZK proof (or argument) with $\ec+ \es + (t-1)\cdot \ez <1$, then one-way functions exist. 
\end{theorem}

In the above case, the techniques of \cite{DBLP:conf/istcs/OstrovskyW93} alone would have resulted in the condition being $\ec+ \es + t\cdot \ez <1$ instead. For constant-round protocols, we improve this condition much more significantly, though in this case we only obtain infinitely-often one-way functions. Below, a \emph{round} refers to one pair of messages in the protocol -- one from the verifier and its response from the prover.

\begin{theorem}[Informally, \cref{the:cr_owf}]
    If $\NP \not \subseteq \mathsf{P/poly}$, and for some constant $k$, every language in $\NP$ has a $k$-round $(\ec, \es, \ez)$-public-coin ZK proof (or argument) with $\ec+ \es + k\cdot\ez <1$, then infinitely-often one-way functions exist. 
\end{theorem}

Our results apply to protocols with computational (weak) zero-knowledge and computational (weak) soundness (i.e., arguments), and thus capture the most general class of such protocols.


\paragraph{Open problems}
Our work leaves open interesting questions around the power of weak zero knowledge systems. We mention some of these below. 

\begin{itemize}
    \item An obvious question is if our analysis can be carried over to the setting of private-coin weak ZK protocols -- the standard ZK to OWF implications hold for such protocols as well, and it is of interest to achieve parity here in the weak ZK setting. 
    \item Our final result works for better parameters but only implies {\em infinitely often} OWFs. This is a limitation of our approach, and it is interesting to improve this to yield standard OWFs. 
    \item A related improvement is to also get improved error parameters for super-constant-round protocols. This will also require new tools or approaches. 
    \item Additionally, it is an exciting problem to consider what other, possibly stronger cryptographic primitives weak ZK or more generally even standard ZK protocols may imply. 
\end{itemize}


\paragraph{Paper outline}
We continue with a technical overview of our results and proofs in \Cref{sec:techovw}. \Cref{sec:prelims} contains our definitions and notation. The first stage of our results are covered in \Cref{sec:aiowf}, which shows how the various kinds of weak ZK protocols we consider yield (variants of) auxiliary input one-way functions. The final implications to one-way functions are shown in \Cref{sec:owf}. 



\subsection{Technical Overview}\label{sec:techovw}

In this section, we provide a high-level overview of the main ideas and techniques behind our results. We start by reviewing the construction of~\cite{DBLP:conf/istcs/OstrovskyW93} with their analysis for non-interactive ZKs. We then introduce our improved construction and outline the ideas that enable improvement. Further, we find that our approach naturally generalizes to a broader setting -- specifically that of public-coin ZK protocols, which we will discuss later. 

Throughout, we use the notation $\cU$ to denote the uniform distribution over strings whose length will be clear from the context.

\paragraph{Non-interactive zero-knowledge}

We first consider non-interactive ZK (NIZK) arguments. A NIZK argument for a language $\cL$ allows a prover to produce a single proof $\pi$ to certify that an input $x\in \cL$ while preserving zero-knowledge. Specifically, given a uniformly random string $r$, also known as the \emph{common reference string}, a polynomial-time prover holding with a valid $\NP$ witness $w$ for $x$ computes a proof $\pi \gets \sfP(x, w;r)$; upon receiving the prover's message, the verifier computes $a \gets \sfV(x; r,\pi)$ to decide whether $x\in \cL$. 

The protocol is required to satisfy completeness and computational soundness, where the errors are correspondingly denoted by $\ec$ and $\es$. Besides these, it also satisfies computational zero-knowledge. In particular, there is a polynomial-time simulator $\Sim$ that on input $x$ generates a distribution $(r,\pi)$ such that no polynomial-time distinguisher can distinguish between this and the $(r,\pi)$ from the actual protocol with advantage greater than the zero-knowledge error $\ez$. For simplicity, we assume that the NIZK protocol has perfect completeness ($\ec=0$) and we have a deterministic verification algorithm $\sfV$.

\paragraph{The Ostrovsky-Wigderson approach}
The key observation underlying~\cite{DBLP:conf/istcs/OstrovskyW93} is that an inverter for the NIZK simulator can be used to construct a distinguisher to decide the language, contradicting its hardness. 

Suppose the language $\cL$ that has the NIZK argument $(\sfP,\sfV,\Sim)$ is worst-case hard. Let $\Sim$ be the simulator that runs on the input $x$ and randomness $\rho$, and outputs the common reference string $r$ and a proof $\pi$. The candidate hard-to-invert function $f_x$ is defined to be
\begin{align*}
    f_x(\rho): \quad&(r,\pi)\gets \Sim(x; \rho)\\ &\text{output }r
\end{align*}
Since the randomness is treated explicitly as part of the input, the above construction is deterministic and therefore the function is well defined. 

Towards a contradiction, assume that there are no auxiliary-input one-way functions. In fact, suppose that there is an adversary $\cA$ that perfectly inverts $f_x$ distributionally -- that is, given $y$, $\cA$ samples a uniformly random pre-image $f_x^{-1}(y)$. It is known that distributional OWFs imply OWFs~\cite{DBLP:conf/focs/ImpagliazzoL89}, so the only loss of generality here is the assumption that the inverter is perfect, but this is not difficult to remove at the cost of an inverse polynomial loss in parameters.

We have that the joint distributions
\begin{equation*}
    (\cA(f_x(\cU)), f_x(\cU)) \approx (\cU, f_x(\cU))
\end{equation*}
 are close, where the left-hand side represents the distribution of first computing $f_x$ on random inputs and applying $\cA$, while the right-hand side denotes the distribution of sampling a random input $r$ and then outputting $(r,f_x(r))$. 

\medskip
The algorithm $\cD$ that decides $\cL$ is as follows. Given input $x$, it invokes $\cA$ to decide whether $x\in \cL$ or not: it samples $r \gets \cU$, runs $\hat{\rho}\gets \cA(r)$, then computes $(\hat{r},\hat{\pi}) \gets \Sim(x;\hat{\rho})$, and accepts if and only if $\sfV(x; r,\hat{\pi})=1$. Note that under our assumptions, we will always have $r=\hat{r}$. We now analyze the performance of $\cD$ in the two possible cases. 
\begin{enumerate}
    \item When $x\in \cL$, completeness implies that when $(r,\pi)$ is generated following the protocol, $\sfV(x;r, \pi) = 1$ holds with probability $1$. Zero-knowledge guarantees that the probability that $\sfV(x;r, \pi) = 1$ when $(r,\pi) \gets \Sim(x; \cU)$ is at least $1 - \ez$. 

    In the algorithm $\cD(x)$, we run $\sfV$ on $(r,\pi)$ generated from $\Sim(x;\cA(\cU))$. When the inverter $\cA$ is run on $r$ sampled from $\Sim(x;\cU)$, the resulting inverse is uniformly random (due to the definition of $f_x$). Again by zero-knowledge, the uniform distribution of $r$ is $\ez$-indistinguishable from the distribution of $r$ sampled by $\Sim(x;\cU)$. So the distribution of $\cA(\cU)$ is also $\ez$-indistinguishable from uniform. This implies that the distribution of $\Sim(x;\cA(\cU))$ is $\ez$-indistinguishable from $\Sim(x;\cU)$. Altogether, we have
    \begin{equation*}
        \Prob{}{\cD(x) = 1} \ge 1 - 2\ez. 
    \end{equation*}
    \item When $x\notin \cL$, soundness ensures that for any efficient method of generating $\pi$ for random $r$, the probability that $\sfV$ accepts is bounded by $\es$, and so
    \begin{equation*}
        \Prob{}{\cD(x) = 1} \le \epsilon_s. 
    \end{equation*}
\end{enumerate}

Consequently, if $\es < 1-2\ez$, the inverter $\cA$ can be used to determine whether $x$ is in $\cL$. If such an inverter works for every $x$, then we can decide the language, contradicting its hardness. Thus, the family of functions $\set{f_x}$ must be a family of auxiliary-input one-way functions.

\paragraph{Improving the condition on errors}
More recently, such analysis has seen further progress. In particular, this bound was improved by \cite{DBLP:conf/crypto/ChakrabortyHK25}, where it was shown that $\sqrt{\epsilon_s} + \ez < 1$ suffices. This was shown with sophisticated arguments using a one-sided version of universal approximation, where the inverter is used to estimate the probabilities of certain outputs.


Our starting point is the observation that with rather simple but careful arguments, it is feasible to relax this bound to $\epsilon_s + \epsilon_z<1$, which is the most general it can be. We avoid paying for the zero-knowledge error twice by involving the verification procedure inside the candidate one-way function. 
Our one-way function is as follows
\begin{align*}
    f_x(\rho):\quad &(r,\pi) \gets \Sim(x;\rho)\\
    &a \gets \sfV(x; r,\pi)\\
    &\text{output } (r,a)
\end{align*}
Suppose again that there is a near-perfect polynomial-time inverter $\cA$ for $f_x$ (it need not be a distributional inverter). That is,
\begin{equation*}
    \Prob{\substack{(r,a) \gets f_x(\cU)}}{f_x(\cA(r,a)) = (r,a)} \approx 1,
\end{equation*}
The algorithm $\cD$ for $\cL$ can be constructed as follows. Given input $x$, sample $r$ uniformly at random, compute $\hat{\rho} \gets \cA(r,1)$, and accept if and only if this is a valid pre-image of $(r,1)$ under $f_x$ -- that is, iff $f_x(\hat{\rho}) = (r,1)$.
\begin{enumerate}
    \item For $x\in \cL$, consider the following procedure $g(r,\pi)$: it computes $a \gets \sfV(x;r,\pi)$ and outputs $(r,a)$. When $(r,\pi)$ is sampled from the simulator $\Sim(x;\cU)$, the distribution of $g(r,\pi)$ is equivalent to $f_x(\cU)$. When $(r,\pi)$ follows the protocol view, the distribution of $g(r,\pi)$ is the same as $(r,1)$ in $\cD$. So by ZK and the data processing inequality, these distributions are $\ez$-indistinguishable. 
    
    Further, when $(r,\pi)$ is sampled from the simulator, the inverter $\cA$ will almost always find a valid pre-image of $(r,a)$. Therefore, given $(r,1)$ as in $\cD$, it finds a valid pre-image with overall probability at least $\approx 1-\ez$.
    \begin{equation*}
        \Prob{}{\cD(x) = 1} = \Prob{\substack{r\gets \cU}}{f_x(\cA(r,1)) = (r,1)}  \ge 1 - \ez.
    \end{equation*}
    \item For $x\notin \cL$, whenever a valid pre-image $\hat{\rho}$ for $(r,1)$ is found, the corresponding $(r,\pi)\gets\Sim(x;\hat{\rho})$ is accepted by $\sfV$. As both $\cA$ and $\Sim$ are efficient algorithms, soundness guarantees that a valid inverse cannot be found with probability more than $\es$. So
    \begin{equation*}
        \Prob{}{\cD(x) = 1} \le \es. 
    \end{equation*}
\end{enumerate}
Therefore, following the same remaining arguments as earlier, $\es+\ez < 1$ suffices to imply the existence of auxiliary-input one-way functions.

The key idea behind our improvement is that we implicitly utilize the verification $\sfV$ while restricting the inverter to output a good pre-image corresponding to a valid proof. Thus as long as the inverter succeeds, there is no need to perform an additional explicit verification. This saves us the extra zero-knowledge error penalty. 

\paragraph{The distinguisher: an alternate view} Our analysis reveals that in both the \cite{DBLP:conf/istcs/OstrovskyW93} construction and our new one, these distinguishers can be regarded as providing efficient simulations of the protocol, where the original honest prover algorithm is replaced with an efficient algorithm without the witness and $\sfV$ serves to certify correctness. For example, the Ostrovsky-Wigderson distinguisher admits an equivalent formulation as a protocol between a prover $\tsfP$ and the verifier $\sfV$. The prover $\tsfP$ performs: on randomness $r$, find the message $\Tilde{\pi}$ corresponding to the simulator randomness $\Tilde{\rho}$, where $\Tilde{\rho}$ is found efficiently by the inverter. Our construction yields a similar prover as well. The only difference is the modification to the inverter since the correctness of a pre-image ensures the validity of corresponding proof.

\paragraph{Multiple-round public-coin zero-knowledge}
Our approach naturally extends to the setting of multiple-round public-coin zero-knowledge protocols. Such a protocol allows interaction between the prover and the verifier, where all of the verifier's communication consists of uniform random bits. In particular, all its randomness is public and available to the prover. 

For the sake of exposition, we assume below perfect completeness with soundness error $\es$ and zero-knowledge error $\ez$. By adapting our approach for the NIZK case, we obtain that the existence a $k$-round public-coin ZK protocol for a hard language $\cL$ with parameters $\es + (2k-1) \ez< 1$ implies the existence of one-way functions. We use the term \emph{round} to denote one interaction where the verifier sends a random challenge and the prover responds with a proof message. 

\paragraph{Better bounds for constant rounds}
Moreover, the bound can be further improved to $\es + k\cdot \ez <1$ using a more sophisticated argument when $k$ is only a constant. 

For illustration, we focus here on the simplest setting in which the language $\cL$ admits a two-round public-coin ZK protocol. 
Here both parties share a common input $x$ while the honest prover additionally holds a witness $w$. In the first round, the verifier first samples a random string $r_1$ and the prover replies with a message $\pi_1 \gets \sfP_1(x, w; r_1)$. In the second round, the verifier sends another random string $r_2$ and the prover responds with $\pi_2 \gets \sfP(x, w;r_1,\pi_1,r_2)$. Finally, the verifier checks the transcript by computing $0/1 \gets \sfV(x; r_1, \pi_1, r_2, \pi_2)$, where $1$ denotes acceptance. 

The zero-knowledge condition implies the existence of a simulator $\Sim$ that on input $x$ and randomness $\rho$, outputs a transcript $(r_1, \pi_1, r_2, \pi_2)$; for any $x\in \cL$ and a valid witness $w$, $\Sim(x)$ is $\ez$-indistinguishable from the protocol view. 
\begin{equation*}
    \Delta(\Sim(x); \prot{\sfP}{\sfV}(x,w)) \le \epsilon_z. 
\end{equation*}

\newcommand{\perf}{\mathsf{Perf}}
\newcommand{\Succ}{\mathsf{Succ}}

\paragraph{A recursive approach}
In the following, we will follow a recursive approach to construct our candidate one-way functions. More precisely, we begin by defining a function, and any algorithm that breaks the one-wayness of this function will be incorporated into the construction of a second function. If the second one is not one-way either, then the resulting inverters can be combined to decide the language. Accordingly, we first define a function $f_{2,x}$ as follows:
\begin{align*}
    f_{2,x}(\rho): \quad &(r_1,\pi_1, r_2,\pi_2) \gets \Sim(x; \rho)\\
    &a \gets \sfV(x; r_1, \pi_1, r_2, \pi_2)\\
    &\text{output }(r_1,\pi_1, r_2, a)
\end{align*}
Assume that there is a poly-time algorithm $\cA_2$ that inverts $f_{2,x}$
\begin{equation}
\label{eq:ov_a2_succ}
    \Prob{(r_1,\pi_1, r_2, a) \gets f_{2,x}(\cU)}{f_{2,x}(\cA_2(x;r_1,\pi_1, r_2, a)) = (r_1,\pi_1, r_2, a)} \approx 1.
\end{equation}
As in the NIZK case, the deciding procedure is essentially equivalent to efficiently simulating the protocol by constructing an efficient prover $\Tilde{\sfP}$: this queries the inverter $\cA$ on input $(r,1)$ and generates the prover's message $\pi$ accordingly. For the second round of the protocol, we define the strategy $\tsfP_2$ similarly. 
\begin{align*}
    \tsfP_2(x;r_1,\pi_1,r_2): \quad & \Tilde{\rho} \gets \cA_2(x;r_1,\pi_1, r_2, 1)\\
    & (\Tilde{r}_1, \Tilde{\pi}_1, \Tilde{r}_2, \Tilde{\pi}_2) \gets \Sim(x; \Tilde{\rho})\\
    &\text{output }\Tilde{\pi}_2
\end{align*}
We measure the performance of $\tsfP_2$ formally by the following quantity:
\begin{align*}
    \Succ_3(x; r_1,\pi_1, r_2) = \mathsf{1}[f_{2,x}(\cA_2(x; r_1,\pi_1,r_2,1)) = (r_1,\pi_1,r_2,1)]
\end{align*}
where $\cA_2$ is assumed to be deterministic for simplicity. The value of $\Succ_3$ represents the probability that $\cA_2$ finds a valid pre-image with respect to $f_{2,x}$. This provides a lower bound for the acceptance probability of the protocol conditioned on the first 3 messages being $(r_1,\pi_1, r_2)$, since $\Tilde{\pi}_2$ is valid as long as $\cA_2$ finds a valid pre-image of $f_{2,x}$.
\begin{equation*}
    f_{2,x}(\Tilde{\rho})= (r_1,\pi_1, r_2, 1) \Rightarrow \sfV(x; r_1,\pi_1, r_2, \Tilde{\pi}_2) = 1. 
\end{equation*}

Since the verifier samples $r_2$ uniformly, we can equivalently lower bound the success probability of this prover by the following expression (where the first two messages are $(r_1,\pi_1)$ and $\tsfP_2$ is the second round strategy of the prover):

\begin{equation*}
    \Succ_2(x; r_1,\pi_1) = \Expec{r_2 \gets \cU}{\Succ_3(x; r_1,\pi_1, r_2)} 
\end{equation*}

Now consider the protocol $\prot{\sfP}{\sfV}$ that replaces the prover algorithm in the second round with $\tsfP_2$. We can establish an upper bound for the completeness error of this modified protocol. 

Denote by $D_{2,x}$ the distribution that samples $(r_1, \pi_1, r_2,\pi_2)$ from the protocol $\prot{\sfP}{\sfV}(x, w)$, sets $a \gets \sfV(x; r_1,\pi_1, r_2,\pi_2)$ and outputs $(r_1,\pi_1, r_2, a)$. 
By the perfect completeness of $\prot{\sfP_w}{\sfV}$, $a = 1$ always holds. 
It follows by the data processing inequality that, for $x\in \cL$
\begin{equation*}
    \Delta(f_{2,x}(\cU); D_{2,x}) \le \Delta(\Sim(x); \prot{\sfP}{\sfV}(x)) \le \ez. 
\end{equation*}
Combining above with our assumption (\ref{eq:ov_a2_succ}), the acceptance probability is at least
\begin{align}
\label{eq:ov_bd_succ2}
    &\Expec{\substack{r_1\gets \cU\\ \pi_1 \gets \sfP_1(x, w;r_1)}}{\Succ_2(x;r_1,\pi_1)}\notag\\ &= \Prob{(r_1,\pi_1, r_2, a) \gets D_{2,x}}{f_{2,x}(\cA_2(x; r_1,\pi_1,r_2,a)) = (r_1,\pi_1,r_2,a)} \notag\\
    &\ge  \Prob{(r_1,\pi_1, r_2, a) \gets f_{2,x}(\cU)}{f_{2,x}(\cA_2(x; r_1,\pi_1, r_2, a)) = (r_1,\pi_1, r_2, a)} - \ez \notag\\
    &\ge 1 - \ez.
\end{align}

The above implies that when the prover applies $\sfP_1(x,w)$ and $\tsfP_2(x)$ in each round, respectively, the acceptance probability is at least $1-\ez$ when $x\in \cL$. 
However, because the witness for $\sfP_1$ is unavailable to our deciding strategy, we need a different polynomial-time algorithm as a replacement to obtain an efficient simulation for the protocol. 

We now turn to building our first-round strategy. 
A natural attempt is to define a new function $f_{1,x}(\rho)$ analogously to $f_{2,x}$ and apply the inverter to produce the message $\pi_1$, where $f_{1,x}(\rho)$ can be defined as follows: sample $(r_1,\pi_1,r_2,\pi_2)\gets \Sim(x;\rho)$, $a\gets \sfV(x; r_1,\pi_1, r_2,\pi_2)$, and output $(r_1,a)$ (or only output $r_1$). 
However, our analysis of the resulting prover algorithms shows that the best achievable result is $\es + 3\ez <1$, which matches $\es + (2k-1)\ez <1$ for $k=2$. We refer the reader to \cref{sec:pczk_aiowf} for the details.

Instead, we introduce a new construction of $f_{1,x}$ that leads to an improved bound. 
After fixing $\tsfP_2$, our objective is to find an efficient way $\tsfP_1$ to generate $\pi_1$ that maximizes the value $\Succ_2(x; r_1,\pi_1)$ optimally for $x\in \cL$. 
The value of $\Succ_2$ can be efficiently approximated up to any arbitrary inverse-polynomial error due to our assumption on $\cA_2$ and standard concentration arguments. For convenience, we ignore the accuracy issue here and assume that we are able to compute $\Succ_2$ precisely in polynomial-time. 
The key is to instead have the new function compute $\Succ_2$ for the relevant partial transcript. More precisely, we define: 
\begin{align*}
    f_{1,x}(\rho):\quad &(r_1, \pi_1, r_2, \pi_2) \gets \Sim(x; \rho)\\
    &\text{output } (r_1, \Succ_2(x;r_1,\pi_1))
\end{align*}
Suppose now that there is an efficient algorithm $\cA_1$ 
\begin{equation*}
    \Prob{(r_1,a) \gets f_{1,x}(\cU)}{f_{1,x}(\cA_1(x; r_1, a)) = (r_1, a)} \approx 1. 
\end{equation*}
We are now ready to formally specify $\tsfP_1$, which works as follows: 
\begin{align*}
    \tsfP_1(x;r_1):\quad &\Tilde{a} \gets \arg \max_a \set{a\cdot \mathbf{1}[f_{1,x}(\cA_1(x; r_1, a)) = (r_1, a)}] \\
    &\Tilde{\rho} \gets \cA_1(x; r_1, \Tilde{a})\\
    &\Tilde{\pi}_1 \gets \Sim(\Tilde{\rho})\\
    &\text{output }\Tilde{\pi}_1
\end{align*}
where we take the support size of $a$ to be polynomial, enabling us to efficiently iterate over all possible values and find the maximum $\Tilde{a}$. 
In fact, in general one can efficiently find an approximate maximum value instead, which suffices.

On input $(x; r_1)$, the prover $\tsfP_1$ goes through all possible values of $a$, and selects the highest one on which $\cA_1$ inverts $f_{1,x}$ successfully. Note that when $\cA_1$ finds a correct pre-image of $(r_1,a)$, it implies that there exists a consistent $\pi_1$ with value $\Succ_2(x; r_1,\pi_1) = a$; so $\tsfP_1$ takes the maximal $a$ and outputs its associated prover message.

With the construction of $\tsfP_1$ and $\tsfP_2$, it remains to analyze the performance of the resulting protocol $\prot{\tsfP}{\sfV}$. More specifically, we are interested in
\begin{equation}
\label{eq:ov_acc}
    \Expec{\substack{r_1 \gets \cU, \pi_1 \gets \tsfP_1(x; r_1)\\ r_2 \gets \cU, \pi_2 \gets \tsfP_2(x; r_1,\pi_1, r_2)}}{\sfV(x; r_1,\pi_1, r_2, \pi_2)}
\end{equation}
Recall that 
\begin{align}
\label{eq:ov_max}
    (\ref{eq:ov_acc}) &\ge  {\Expec{\substack{r_1 \gets \cU \\ \pi_1 \gets \tsfP_1(x; r_1)}}{\Succ_2(x; r_1, \pi_1)}}\notag\\ &= \Expec{r_1 \gets \cU}{\max_a \set{a\cdot \mathbf{1}[f_{1,x}(\cA_1(x; r_1, a)) = (r_1, a)]}}
\end{align}
where the last equality is ensured since by definition of $\tsfP_1(x; r_1)$, the $\Succ_2$ estimate $a$ associated with $\pi_1$ is the largest such value with respect to which $\cA_1$ can succeed, i.e., 

\begin{equation*}
    \Succ_2(x; r_1,\pi_1) = \max_a \set{a\cdot \mathbf{1}[f_{1,x}(\cA_1(x; r_1, a)) = (r_1, a)]}. 
\end{equation*}
Now observe that, for any $r_1$ and any particular $a^*$
\begin{equation}
\label{eq:ov_max_bd}
    \max_a \set{a\cdot \mathbf{1}[f_{1,x}(\cA_1(x;r_1, a)) = (r_1, a)]} \ge a^* \cdot \mathbf{1}[f_{1,x}(\cA_1(x;r_1, a^*)) = (r_1, a^*)]. 
\end{equation}
Let $a^*$ above be sampled as $a^*\gets \Succ_2(x; r_1, \pi_1)$ where $\pi_1 \gets \sfP_1(x, w; r_1)$. We now investigate the new completeness error to obtain a lower bound for (\ref{eq:ov_max}) when $x\in \cL$. 
Note that the value of $\Succ_2$ falls in the range $[0,1]$, thus we have
\begin{align}
\label{eq:ov_lb_succ}
    (\ref{eq:ov_max}) &\ge \Expec{\substack{r_1 \gets \cU, \pi_1 \gets \sfP_1(x, w; r_1)\\ a\gets \Succ_2(x; r_1, \pi_1)}}{a \cdot \mathbf{1}[f_{1,x}(\cA_1(x; r_1, a)) = (r_1, a)]}\notag\\ 
    &\ge  \Expec{\substack{r_1 \gets \cU, \pi_1 \gets \sfP_1(x, w; r_1)\\ a\gets \Succ_2(x; r_1, \pi_1)}}{a} - \Expec{\substack{r_1 \gets \cU, \pi_1 \gets \sfP_1(x, w; r_1)\\ a\gets \Succ_2(x; r_1, \pi_1)}}{\mathbf{1}[f_{1,x}(\cA_1(x; r_1, a)) \neq (r_1, a)]}\notag\\
    &\ge 1 - \ez - \Prob{\substack{r_1 \gets \cU, \pi_1 \gets \sfP_1(x, w; r_1)\\ a\gets \Succ_2(x; r_1, \pi_1)}}{f_{1,x}(\cA_1(x; r_1, a)) \neq (r_1, a)},
\end{align}
where the first inequality is obtained by the observation (\ref{eq:ov_max_bd}) and the last inequality holds by (\ref{eq:ov_bd_succ2}). 
Now define the distribution $D_{1,x}$: sample $r_1 \gets \cU$ and $\pi_1 \gets \sfP_1(x, w; r_1)$, and finally output $(r_1, \Succ_2(x;r_1,\pi_1))$. 

For $x\in \cL$, by the data processing inequality and zero-knowledge condition  
\begin{equation*}
    \Delta(f_{1,x}(\cU); D_{1,x}) \le \Delta(\Sim(x); \prot{\sfP_w}{\sfV}(x)) \le \ez. 
\end{equation*}
Since we assume that $\cA_1$ inverts the function almost perfectly, we get 
\begin{align}
\label{eq:ov_bd_fail_A1}
    &\Prob{\substack{r_1 \gets \cU, \pi_1 \gets \sfP_1(x,w; r_1)\notag\\ a\gets \Succ_2(x; r_1, \pi_1)}}{f_{1,x}(\cA_1(x; r_1, a)) \neq (r_1, a)}\\ &\le \Prob{\substack{(r_1, a) \gets f_{1,x}(\cU) }}{f_{1,x}(\cA_1(x; r_1, a)) \neq (r_1, a)} + \ez \notag\\ &\le \ez. 
\end{align}
Therefore, by combining (\ref{eq:ov_acc}), (\ref{eq:ov_lb_succ}) and (\ref{eq:ov_bd_fail_A1}), we derive that for $x\in \cL$
\begin{equation*}
    \Expec{\substack {r_1 \gets \cU, \pi_1 \gets \tsfP_1(x; r_1)\\ r_2 \gets \cU, \pi_2 \gets \tsfP_2(x; r_1,\pi_1, r_2)}}{\sfV(x; r_1,\pi_1, r_2, \pi_2)} \ge 1 - 2\ez. 
\end{equation*}

On the other hand, soundness of the original protocol $\prot{\sfP_w}{\sfV}$ ensures that, for $x\notin \cL$
\begin{equation*}
    \Expec{\substack {r_1 \gets \cU, \pi_1 \gets \tsfP_1(x; r_1)\\ r_2 \gets \cU, \pi_2 \gets \tsfP_2(x; r_1,\pi_1, r_2)}}{\sfV(x; r_1,\pi_1, r_2, \pi_2)} \le \es. 
\end{equation*}

When $\es + 2\ez$ is noticeably less than $1$, we thus obtain an efficient algorithm that decides if $x\in \cL$, which yields the desired result. 
\rnote{looks okay overall, I think we can add a small remark here explaining why ioOWFs and redirecting to the larger explanation in the main text.}

One catch here is that while this approach is conceptually complete, it only gives us a construction of an infinitely often one-way function. This is because in our approach, we will require that the adversaries $\cA_1$ and $\cA_2$ must both succeed in their inversion so that we successfully decide on an instance. This means that the sequence of input lengths on which our assumed inverters work must overlap when we aim for a contradiction - and the formal negation for this only implies infinitely-often hardness. See \Cref{sec:const-proof} and \Cref{rem:ioOwf} for details. 

\paragraph{Handling randomized verification in NIZKs}

We note that it requires non-trivial techniques to derive a similar bound for the errors of NIZK protocols when the verification step $\sfV$ is randomized.
Note that it is feasible to deal with the randomized case by adapting the approach in the multi-round public-coin case. 
We briefly describe the construction for the NIZK case below. 
The construction of candidate function is analogous to $f_{1,x}$ defined before. In particular, 
\begin{align*}
    f_x(\rho):\quad& (r,\pi) \gets \Sim(x; \rho)\\
            & \text{output }(r, \mathbb{E}[\sfV(x; r,\pi)])
\end{align*}
We also assume for now that one can compute this expected value deterministically in polynomial time. 
Given an efficient inverter $\cA$, we can define an efficient prover strategy that runs (without any witness) as follows:
\begin{align*}
    \tsfP(x;r):\quad &\Tilde{a} \gets \arg \max_a \set{a\cdot \mathbf{1}[f_{x}(\cA(x; r, a)) = (r, a)}] \\
    &\Tilde{\rho} \gets \cA(x; r, \Tilde{a})\\
    &\Tilde{\pi} \gets \Sim(x; \Tilde{\rho})\\
    &\text{output }\Tilde{\pi}
\end{align*}
From a similar argument, for $x\in \cL$, the acceptance probability of $\prot{\tsfP}{\sfV}$ is at least
\begin{align*}
    &\Expec{\substack{r_1 \gets \cU, \pi \gets \sfP(x, w; r)\\ a\gets \mathbb{E}[\sfV(x; r, \pi)]}}{a \cdot \mathbf{1}[f_{x}(\cA(x; r, a)) = (r, a)]}\\ &\ge 1 - \Prob{\substack{r \gets \cU, \pi \gets \sfP(x, w; r)\\ a\gets \mathbb{E}[\sfV(x; r, \pi)]}}{f_{x}(\cA(x; r, a)) \neq (r, a)}\\ &\ge 1 - \ez. 
\end{align*}
For $x\notin \cL$, soundness holds with error probability $\es$. Therefore, we conclude that the condition $\es+ \ez<1$ suffices when the verification is randomized as well. 

%% file: pre.tex
\section{Preliminaries}\label{sec:prelims}

\paragraph{Notations}
Denote by $\cU_\ell$ the uniform distribution over the length-$\ell$ binary strings $\zo^\ell$. For a language $\cL$, let $\cL_n = \cL \cap \zo^n$ be the set of all the length-$n$ strings in the language. 
For $k\in \Nat$, denote $[k] = \set{1,\dots, k}$. 
We define a distribution ensemble $\mathcal{X} = \set{X_n}_{n \in \mathbb{N}}$ to be a collection of distributions where each $X_n$ is defined over $\set{0,1}^{m(n)}$ where $m: \mathbb{N} \rightarrow \mathbb{N}$ is some arithmetic function. 
We say a function $\nu$ is negligible if for every polynomial $p$ there exists an $n_0 \in \mathbb{N}$ such that for all $n \geq n_0$, $\mu(n) < \frac{1}{p(n)}$. Similarly, we call a function $\mu$ noticeable if there exists a polynomial $p$ and $n_0 \in \mathbb{N}$ such that for all $n \geq n_0$, $\mu(n) \geq \frac{1}{p(n)}$.
    
For $\epsilon_1 = \epsilon_1(n) , \epsilon_2= \epsilon_2(n)$, we say $\epsilon_1<_n\epsilon_2$ (or $\epsilon_2>_n \epsilon_1$) to represent the noticeable gap between $\epsilon_1$ and $\epsilon_2$, if there exists a polynomial $p$ such that for all sufficiently large $n\in \Nat$, we have
\begin{equation*}
    \epsilon_1(n) + \frac{1}{p(n)} < \epsilon_2(n), 
\end{equation*}
which implies asymptotically, there is an inverse-polynomial gap between $\epsilon_1$ and $\epsilon_2$. 

We use the Hoeffding bound, stated as follows. 
\begin{lemma}[Hoeffding's inequality]
    Suppose we have independent random variables $X_1, \dots, X_q$ with support $[0,1]$. Let $X =\sum_{i\in [q]} X_i$, then the following holds
    \begin{equation*}
        \Prob{}{\abs{X - \mathbb{E}{X} }> t\cdot q} \le 2e^{-2t^2q}.
    \end{equation*}
\end{lemma}

\subsection{Indistinguishability}

Assume that we have two distributions $X$ and $Y$ defined over a common universe $U$. We first consider statistical indistinguishability. 

\begin{definition}[Statistical Distance]
    The {\em statistical distance} between distributions $X$ and $Y$ (defined over the support of $X$, denoted by $\mathsf{Supp}{(X)}$) is defined as $$\Delta_s(X;Y) = \frac{1}{2} \sum_{u \in \mathsf{Supp}(X)}  \big|\Pr[X=u] - \Pr[Y=u] \big|.$$
\end{definition}

\newcommand{\cZ}{\mathcal{Z}}

The following properties hold. 

\begin{lemma}[Triangle Inequality]
    For any three distributions $X$, $Y$ and $Z$, we have  
    $$\Delta_s(X; Z) \le \Delta_s(X;Y) + \Delta_s(Y;Z).$$
\end{lemma}

\begin{lemma}[Data Processing Inequality]
\label{lem:dpi}
    For any two probability distributions $X$, $Y$ (on a common universe $U$), and any (possibly randomized) process $f$, we have $$\Delta_s(f(X); f(Y)) \leq \Delta_s(X; Y).$$
\end{lemma}

Suppose that we have a given distinguishing algorithm $D$ to distinguish between $X$ and $Y$, taking inputs in $U$ and outputting a bit to indicate whether it identifies a given input as being sampled from $X$ or $Y$. Define the {\em distinguishing advantage} $\mathsf{Adv}_D(X,Y)$ of $D$ as: 
\begin{equation*}
  \mathsf{Adv}_D(X;Y) = \abs{ \Prob{x\gets X}{D(x) = 1} - \Prob{y \gets Y}{D(y) = 1}}. 
\end{equation*}
In fact, the statistical distance implicitly provides an upper bound on the advantage that any distinguisher can obtain, that is 
\begin{equation*}
    \Delta_s(X;Y) = \max_D \mathsf{Adv}_D(X;Y)
\end{equation*}
where $D$ can be any possible algorithm. Usually, for a constant $\epsilon$, when $\Delta_s(X; Y) < \epsilon$, $X$ and $Y$ are said to be $\epsilon$-statistically indistinguishable.  
Next, we formally define the statistical indistinguishability asymptotically. 

\newcommand{\cX}{\mathcal{X}}
\newcommand{\cY}{\mathcal{Y}}

\begin{definition}[Statistical Indistinguishability]
     Consider a function $\epsilon: \Nat \rightarrow [0,1]$, and distribution ensembles $\mathcal{X} = \set{X_n}_{n \in \Nat}$ and $\mathcal{Y} = \set{Y_n}_{n \in \Nat}$. We say that $\cX$ and $\cY$ are $\epsilon$-\emph{statistically indistinguishable} (or have statistical distance at most $\epsilon$), denoted by $\Delta_s(\cX; \cY) \leq \epsilon$, if for all $n \in \Nat$, we have  $\Delta_s(X_n; Y_n) \leq \epsilon(n)$.
\end{definition}

Notice that we do not impose any computational constraint on the distinguisher above, thus the statistical notion implies that even computationally unbounded algorithms cannot achieve an advantage better than $\epsilon$. It is also natural to restrict attention to polynomial-time algorithms. Next, we define computational indistinguishability. 


\begin{definition}[Computational Indistinguishability]
    Consider a function $\epsilon: \Nat \rightarrow [0,1]$, and distribution ensembles $\mathcal{X} = \set{X_n}_{n \in \mathbb{N}}$ and $\mathcal{Y} = \set{Y_n}_{n \in \mathbb{N}}$. If for every non-uniform probabilistic polynomial-time algorithm $D$, there is an $n_D\in\Nat$ such that for all $n\geq n_D$ we have 
    \begin{equation*}
        \mathsf{Adv}_D(X_n; Y_n) \le \epsilon(n),
    \end{equation*}
    then we say that $\mathcal{X}$ and $\mathcal{Y}$ are $\epsilon$-\emph{computationally indistinguishable} (with respect to polynomial-time algorithms). We denote this by
    \begin{equation*}
        \Delta_c(\cX; \cY) \le \epsilon. 
    \end{equation*}
\end{definition}
    
In the course of our technical arguments, for the sake of simplicity, we often make statements of the form $\Delta_c(X_n; Y_n) \le \epsilon(n)$. These statements and their implications are to be interpreted in the asymptotic sense, as holding for all large enough $n$ rather than for all $n$.
    


The definition ensures that $\Delta_c(\cX; \cY)\le \Delta_s(\cX; \cY)$.
Versions of the triangle inequality and data processing inequality hold computational indistinguishability as well. The former is easily implied by essentially a hybrid argument. We state it formally as follows. 

\begin{lemma}[Triangle Inequality]
    For any three distribution ensembles $\cX$, $\cY$, $\cZ$, we have 
    \begin{equation*}
        \Delta_c(\cX; \cY)\le \epsilon_1, \Delta_c(\cY; \cZ) \le \epsilon_2 \Rightarrow \Delta_c(\cX; \cZ) \le \epsilon_1 + \epsilon_2. 
    \end{equation*}
\end{lemma}


\begin{lemma}[Data Processing Inequality]
    For any distribution ensembles $\cX, \cY$ and any probabilistic polynomial-time procedure $f$, we have
    \begin{equation*}
        \Delta_c(\cX; \cY) \le \epsilon \Rightarrow \Delta_c(f(\cX); f(\cY)) \le \epsilon. 
    \end{equation*}
\end{lemma}
This is a simple consequence of the observation that any such efficient function $f$ can simply be run on top of samples from $\cX$ or $\cY$ and then fed into a distinguisher for $f(\cX)$ and $f(\cY)$. 

We slightly abuse the notation by writing $ \Delta_c(\cX; \cZ) \le \Delta_c(\cX; \cY)+ \Delta_c(\cY; \cZ)$ and $\Delta_c(f(\cX); f(\cY))  \le \Delta_c(\cX; \cY)$, by which we mean the properties defined above. 

\subsection{Circuits and Oracles}

We define notions of circuits and functions computed with respect to certain oracles. 

\begin{definition}[Oracle-Aided Circuits and Algorithms]
    Let $\mathcal{O}: \set{0,1}^* \rightarrow \set{0,1}^*$ be an oracle (an arbitrary function). An oracle circuit (or algorithm) $C$ with respect to $\mathcal{O}$, denoted by $C^\mathcal{O}$ is a circuit (or algorithm) where in addition to standard operations, $C$ also has oracle gates (or oracle operations) where it can make a query to $\mathcal{O}$, and expect its output as response. 
\end{definition}

\begin{definition}[Oracle-Aided Functions]
    Let $\mathcal{O}: \set{0,1}^* \rightarrow \set{0,1}^*$ be an oracle (an arbitrary function). An oracle aided function $f$ with respect to $\mathcal{O}$, denoted $f^\mathcal{O}$, is a function computable by a deterministic oracle-aided algorithm $A^\mathcal{O}$. 
\end{definition}

\begin{remark}
    When $\cO$ is a randomized algorithm, we view the oracle gates for $\cO$ as deterministic ones that take an additional randomness as input. 
\end{remark}

\subsection{One-Way Functions}

In the following, we present the definition of one-way functions, along with several weaker variants, which will serve as intermediate steps in our later proofs. 
\begin{definition}[One-Way Function, OWF]
    For $m_1, m_2$ being polynomials, a function family $\cF = \{f_n: \set{0,1}^{m_1(n)} \rightarrow \set{0,1}^{m_2(n)}\}_{n \in \Nat}$ is said to be a {\em One-Way Function (OWF)} if $\cF$ is efficiently computable and for every non-uniform PPT algorithm $\cA$ there is a negligible function $\nu(\cdot)$ such that for all large enough $n\in \Nat$, $$\Pr_{ x \gets \cU_{m_1(n)}}\big[ f_n(\cA(f_n(x))) = f_n(x) \big] \le \nu(n).$$ 
\end{definition}

\begin{definition}[Weak One-Way Function]
    For $m_1,m_2$ being polynomials, a function family $\cF = \{f_n: \set{0,1}^{m_1(n)} \rightarrow \set{0,1}^{m_2(n)}\}_{n \in \Nat}$ is said to be a {\em Weak One-Way Function} if $\cF$ is efficiently computable and for every non-uniform PPT algorithm $\cA$, there is a polynomial $p$ such that for all large enough $n\in \Nat$, $$\Pr_{ x \gets \cU_{m_1(n)}}\big[ f_n(\cA(f_n(x))) = f_n(x) \big] \le 1 - \frac{1}{p(n)}. $$
\end{definition}

\begin{definition}[Distributional One-Way Function, dOWF]
    For $m_1,m_2$ being polynomials, 
    a function family $\cF = \{f_n: \set{0,1}^{m_1(n)} \rightarrow \set{0,1}^{m_2(n)}\}_{n \in \Nat}$ is a {\em Distributional One-Way Function (dOWF)} if $\cF$ is efficiently computable and for every non-uniform PPT algorithm $A$, there is a polynomial $p$ such that for all large enough $n$, the following two distributions: 
    \begin{itemize}
        \item $X_n : \big\{ (x,f(x)): x \gets \cU_{m_1(n)} \big\}$
        \item $Y_n : \big\{ (\cA(f(x)),f(x)): x \gets \cU_{m_1(n)} \big\}$
    \end{itemize}
    satisfy the following $$\Delta_s(X_n; Y_n) > \frac{1}{p(n)}. $$ 
\end{definition}


\begin{remark}
    In our arguments, we will consider adversaries against distributional one-way functions. We will refer to such an adversary $\cA$ as inverting the function in a distributional sense (or distributionally inverting it as shorthand), with deviation say $\frac{1}{q(n)}$ (where $q(\cdot)$ is a polynomial), to mean that $\Delta_s\big( (x,f(x)); (\cA(f(x),f(x))) \big) \le \frac{1}{q(n)}$ for uniformly sampled $x \gets \cU_{m_1(n)}$.  
\end{remark}

\begin{definition}[Auxiliary-Input One-Way Functions, ai-OWF]
    For $m_1,m_2$ being polynomials, 
    a function family $\cF = \{f_a: \set{0,1}^{m_1(\size{a})} \rightarrow \set{0,1}^{m_2(\size{a})}\}_{a \in \set{0,1}^*}$ 
    is said to be an {\em Auxiliary-Input One-Way Function (ai-OWF)} if $\cF$ is efficiently computable and for every non-uniform PPT machine $\cA$ there is a negligible function $\nu(\cdot)$ such that for all large enough $n \in \mathbb{N}$, there exists some $a \in \set{0,1}^n$ such that we have $$\Pr_{ x \gets \cU_{m_1(n)}}\big[ f_a(A(a,f_a(x))) = f_a(x) \big] \le \nu(n).$$
\end{definition}

\begin{remark}
    If in the above definition the string $a$ (for a given $n$) is always fixed to be $0^n$, then the definition collapses to that of a standard one-way function. 
\end{remark}

\begin{remark}
    Similar to the above variant of (standard) OWFs, we can also extend the definitions of weak and distributional one-way functions to the auxiliary input setting in the natural manner. 
\end{remark}

\begin{definition}[Infinitely-Often One-Way Functions, ioOWF]
    For $m_1$, $m_2$ being polynomials, a function family $\cF = \{f_n: \set{0,1}^{m_1(n)} \rightarrow \set{0,1}^{m_2(n)}\}$ is an {\em Infinitely Often One-Way Function (ioOWF)} if $\cF$ is efficiently computable and if for every non-uniform PPT algorithm $A$ there is a negligible function $\nu(\cdot)$ and an infinite set $S_A \subseteq \mathbb{N}$ such that we have $$\Pr_{ r \gets \set{0,1}^{m_1(n)}}\big[ f_n(A(f_n(r))) = f_n(r) \big] \le \nu(n)$$ for all $n \in S_A$.
\end{definition}

\begin{remark}
    When the properties of any of the primitives above hold only for infinitely many $n \in \Nat$, we refer to them as infinitely-often versions. Similarly, infinitely-often versions of complexity classes also be defined -- e.g., $\mathsf{ioP}$ is the set of languages that have deterministic polynomial-time algorithms that are correct on some infinite set of input lengths.
\end{remark}

\begin{remark}
    Similar to above, we can also define weak and distributional infinitely often one-way functions with straightforward modifications.
\end{remark}

\begin{lemma}[\cite{DBLP:conf/focs/ImpagliazzoL89}]
    There is an explicit, efficient transformation from any distributional one-way function to a standard one-way function. 
\end{lemma}

\begin{lemma}[\cite{DBLP:conf/focs/Yao82a}]
    There is an explicit, efficient transformation from any weak one-way function to a standard one-way function. 
\end{lemma}

\begin{remark}
    Both results cited above work as stated for converting (infinitely-often or auxiliary-input) weak or distributional one-way functions to (infinitely-often or auxiliary-input) one-way functions. 
\end{remark}



\subsection{Zero-Knowledge Protocols}

In this section, we formally define zero-knowledge protocols, specifically \emph{non-interactive zero-knowledge} (NIZK) and \emph{public-coin zero-knowledge} proofs and arguments. 

\newcommand{\rel}{\mathcal{R}}

\subsubsection{Non-Interactive Zero-Knowledge}

We describe the non-interactive zero-knowledge proofs or arguments in the \emph{common reference string} model. In particular, there is no interaction between the prover and the verifier. Both parties refer to a common reference string $r$, which is randomly sampled, to run the protocol. For an $\NP$ language $\cL$, on input $x$, the honest prover holds the $\NP$ witness $w$ and generates a message $\pi \gets \sfP(x,w; r)$; then the verifier computes $0/1 \gets \sfV(x; r, \pi)$, where by convention, $1$ denotes the acceptance of the proof. 

\begin{definition}[Non-Interactive Zero-Knowledge, NIZK]
    For $\ec,\es,\ez: \Nat \rightarrow [0,1]$ and a language $\cL\in \NP$, an \emph{$(\ec,\es,\ez)$-Non-Interactive Zero-Knowledge (NIZK) proof} for $\cL$ consists of algorithms $(\gen,\sfP, \sfV)$, where $\gen$ samples the \emph{common reference string} in polynomial time, $\sfV$ is a deterministic polynomial-time verifier and $\sfP$ is a computationally unbounded prover. 
    Let $n$ be the length of the input $x$ and $\rel_\cL$ denote the corresponding $\NP$ relation. 
    The protocol should satisfy the following properties.
    \begin{enumerate}
        \item \emph{Completeness.} For any $x \in \cL_n$ and any witness $w$ such that $(x,w) \in \rel_\cL$
        \begin{equation*}
            \Prob{\substack{r \gets \gen(1^n)\\ \pi \gets \sfP(x, w; r)}}{\sfV(x; r, \pi) = 1} \ge 1 - \epsilon_c(n). 
        \end{equation*}
        \item \emph{Soundness.} For any $x\in\zo^n \setminus \cL_n$ and any prover algorithm $\sfP^*$, the following holds
        \begin{equation*}
            \Prob{\substack{r \gets \gen(1^n) \\ \pi^* \gets \sfP^*(x; r)}}{\sfV(x; r, \pi^*) = 1} \le \epsilon_s(n),
        \end{equation*}
        \item \emph{Computational Zero-Knowledge.} There exists a probabilistic polynomial-time simulator $\Sim$ such that for any $x\in \cL_n$ and $(x,w) \in \rel_\cL$
        \begin{equation*}
            \Delta_c(\Sim(x); \tr(x,w)) \le \ezn,
        \end{equation*}
        where the transcript $\tr(x,w)$ represents the view of the verifier in the protocol with input $x$ and witness $w$ given to the prover, consisting of the common reference string $r$ and the prover's message $\pi$. 

        \medskip
        If the honest prover $\sfP$ is constrained to be computationally efficient, and the soundness condition is required to hold only against computationally efficient provers $\sfP^*$, the protocol is called an \emph{NIZK argument}.
    \end{enumerate}
\end{definition}

\begin{remark}
    The simulator $\Sim(x)$ is randomized on input $x$. 
    When we need a deterministic description, we explicitly expose the random coins and include them as part of the input, which is written as $\Sim(x;\rho)$. Equivalently, $\Sim(x)$ represents the distribution of $\Sim(x;\rho)$ when $\rho$ is drawn uniformly randomly. 
    For convenience, we denote by $\Sim_i(x)$ the simulator $\Sim(x)$ restricted to outputting only the $i$-th message. For example, in the NIZK protocols, $\Sim_1(x)$ only samples the marginal distribution on the common reference string $r$ while $\Sim_2(x)$ simulates the distribution of $\pi$. 
\end{remark}

\begin{remark}
    We have defined the verifier as being a deterministic algorithm in its decision of whether to accept a given execution. While this is not the most general possible notion, by and large known protocols all have a final deterministic verifier step and typically this is the notion considered in most definitions. Nevertheless, we are able to show our results even for the more general notion of randomized verification. The proof of this version of our results is more involved, and the crucial lemma is presented in \Cref{sec:randomizedverifier}. 
\end{remark}



\subsubsection{Public-Coin Zero-Knowledge}

Beyond non-interactive protocols, we study the broader class of public-coin zero-knowledge arguments or proofs. 

\begin{definition}[Public-Coin Zero-Knowledge]
    For $\ec, \es,\ez:\Nat \rightarrow [0,1]$ and a language $\cL$, an \emph{${(\epsilon_c, \epsilon_s, \epsilon_z)}$-public-coin Zero-Knowledge (ZK) proof} for $\cL$ consists of algorithms $(\sfP, \sfV)$, where $\sfV$ is a deterministic polynomial-time verifier and $\sfP$ is a computationally unbounded prover. In an execution of the protocol $\sfV$ is given as input the instance $x$ and $\sfP$ the instance $x$ and a witness $w$. The protocol should satisfy the following.
    \begin{enumerate}
        \item \emph{Syntax and notation.} In each round, first a uniformly random string $r_i$ of pre-specified length is sampled and sent to the prover, and the prover responds with a message $\pi_i$ computed as $\pi_i \gets \sfP_i(x, w; r_1, \pi_1, \dots, r_i)$. At the end, the verifier decides whether to accept the transcript $(r_1,\pi_1, \dots, r_k, \pi_k)$ by computing $0/1 \gets \sfV(x; r_1,\pi_1, \dots, r_k, \pi_k)$. Denote by $\prot{\sfP}{\sfV}(x,w)$ the output of the protocol on a common input $x$ and a witness $w$ (held only by the prover), and $\tr(x,w)$ represents the transcript of the execution, consisting of $(r_1,\pi_1,\dots,r_k,\pi_k)$. Here, $k$ is the number of rounds of the protocol. 
        \item \emph{Completeness.} For any $x\in \cL_n$, for any witness $w$ that $(x,w) \in \rel_\cL$, 
        \begin{equation*}
            \Prob{}{\prot{\sfP}{\sfV}(x,w) = 1} \ge 1 - \ecn.
        \end{equation*}
        \item \emph{Soundness.} For any $x \in \zo^n \setminus \cL_n$, for any prover $\sfP^*$ 
        \begin{equation*}
            \Prob{}{\prot{\sfP^*}{\sfV}(x) = 1} \le \esn. 
        \end{equation*}
        \item {\em Computational Zero-knowledge}: There exists a probabilistic polynomial-time simulator $\Sim$, such that for any $x\in \cL_n$ and $(x,w) \in \rel_\cL$
        \begin{equation*}
            \Delta_c(\Sim(x); \tr(x,w)) \le \ezn. 
        \end{equation*}

        \medskip
        If the honest prover $\sfP$ is constrained to be computationally efficient, and the soundness condition is required to hold only against computationally efficient provers $\sfP^*$, the protocol is called a \emph{ZK argument}.
    \end{enumerate}
\end{definition}

\begin{remark}
    We often think of the public coins $r_i$'s as being sent by the verifier to the prover. We consider the process of the verifier sending a randomness and the prover replying with a message as one round in the protocol, where each round contains two messages. For convenience in notation, when the first message in the protocol is from the prover, we sometimes pretend that there is an empty message from the verifier before that.
\end{remark}



%% file: aiowfbig.tex
\section{Auxiliary-Input One-Way Functions}
\label{sec:aiowf}

In this section, we show that if a language has a Zero-Knowledge proof or argument with errors satisfying certain conditions, then for any instance, we can define a function such that any inverter for that function can be used to decide the membership of that instance in the language. Worst-case hardness of the language then gives us an auxiliary-input one-way function, which will be used in later sections to construct one-way functions from hard languages that have such proof systems.

We do this for Non-Interactive ZK protocols in \cref{sec:nizk_aiowf} and for general public-coin ZK protocols in \cref{sec:pczk_aiowf}. In \cref{sec:constant_aiowf}, we use additional ideas to improve the range of errors that can be used, at the cost of the proof being non-black-box, and yielding only infinitely often secure OWFs. 

\begin{remark}
    \label{rem:proof-argument}
    In this section, we state the main lemmas for both ZK proofs and ZK arguments. To avoid excessive repetition, in the proofs of these lemmas, we only deal with the case of arguments. It may be verified that these proofs do not rely on the efficiency of the honest prover, and work nearly as is for proof systems as well.
\end{remark}

\subsection{Reductions from NIZK}
\label{sec:nizk_aiowf}

\begin{lemma}
\label{lem:nizk_red}
    For some $\ec, \es, \ez: \Nat \rightarrow [0,1]$, suppose a language $\cL\in \NP$ has an $(\ec, \es, \ez)$-NIZK proof or argument (with deterministic verification). Then there exists a reduction $R$, which is a polynomial-time oracle-aided algorithm, and a polynomial-time computable function family $\cF = \set{f_x}_{x\in \zo^*}$ such that, for any probabilistic polynomial-time algorithm $\cA$, any polynomial $p$, and all large enough $n\in\Nat$,
    \begin{enumerate}
        \item\label{it:nizk_yes} For any $x\in \cL_n$, if $\cA$ inverts $f_x$ with probability at least $(1 -1/{p(n)})$, then
        \begin{equation*}
            \Prob{}{R^{\cA}(x) = 1} \ge 1 - \ecn - \ezn - \frac{1}{p(n)}. 
        \end{equation*}
        \item\label{it:nizk_no} For any $x\in \zo^n \setminus \cL_n$, 
        \begin{equation*}
            \Prob{}{R^{\cA}(x) = 1} \le \esn. 
        \end{equation*}
    \end{enumerate}
\end{lemma}

\begin{proofof}{Lemma \ref{lem:nizk_red}}
\label{pf:nizk_aiowf}  
    Let $\Sim$ be the simulator that satisfies the zero-knowledge requirement, and suppose it uses $\ell$ bits of randomness, where $\ell = \ell(n)$. Let $\cF= \set{f_x}_{x\in \zo^*}$ be a function family where $f_x$ is computed as:
    \begin{align*}
        f_x(\rho):\quad &(r,\pi) \gets \Sim(x; \rho) \\
        &a\gets \sfV(x;r,\pi) \\
        &\text{output }(r,a)
    \end{align*}
    where $\rho$ serves as the randomness of $\Sim$. 
    Assume that $\cA$ is a PPT algorithm for potentially inverting $f_x$'s. 
    Construct a reduction $R$ with oracle access to $\cA$, which works as follows:
    \begin{align*}
        R^\cA(x): \quad &r\gets \gen(1^n)\\
        &\rho \gets \cA(x; (r, 1))\\
        &\text{if } f_x(\rho) = (r,1) \text{ output }1\\
        &\text{else } \text{output }0
    \end{align*}

    \medskip
    We first prove that $R$ and $\cF$ satisfy the condition \ref{it:nizk_yes}. For an asymptotic $n\in\Nat$ and $x\in \cL_n$, suppose that $\cA$ inverts $f_x$ with probability at least $(1 - 1/p(n))$, that is 
    \begin{equation}
    \label{eq:nizk_inv}
        \Prob{(r,a) \gets f_x(\cU_\ell)}{f_x(\cA(x;(r,a))) = (r,a)} \ge 1 - \frac{1}{p(n)}. 
    \end{equation}
    Denote by $D_S$, $D_P$ and $D_I$ the following distributions, with $w$ being any valid witness for $x$ that satisfies $(x,w) \in \rel_\cL$.
    \begin{itemize}
        \item $D_S$: sample $\rho \gets \cU_{\ell}$, $(r,\pi) \gets \Sim(x; \rho)$, $a \gets \sfV(x; r, \pi)$, output $(r, a)$
        \item $D_P$: sample $r\gets \gen(1^n)$, $\pi \gets \sfP(x, w; r)$, $a \gets \sfV(x; r, \pi)$, output $(r,a)$
        \item $D_I$: sample $r \gets \gen(1^n)$, output $(r, 1)$
    \end{itemize}
    By the data processing inequality and zero-knowledge, we have: 
    \begin{equation}
    \label{eq:nizk_ez}
        \Delta_c(D_S; D_P) \le \Delta_c(\Sim(x); \tr(x,w)) \le \ez(n). 
    \end{equation}
    The completeness ensures that:
    \begin{equation*}
        \Prob{(r,a) \gets D_P}{ a = 1} \ge 1 - \ec(n).
    \end{equation*}
    Since the marginal distributions of $r$ induced by $D_P$ and $D_I$ are the same, we have:
    \begin{align}
    \label{eq:nizk_ec}
        \Delta_s(D_P; D_I) &= \Expec{r\gets\gen(1^n)}{\Delta_s(D_I|r;D_P|r)} \nonumber\\
        &= \Expec{r\gets\gen(1^n)}{\Prob{a\gets D_I|r}{a=1} - \Prob{a\gets D_P|r}{a=1}} \nonumber\\
        &= 1 - \Prob{(r,a) \gets D_P}{ a = 1} \nonumber\\ &\le \ec(n)
    \end{align}
    By the triangle inequality, the distance between $D_S$ and $D_I$ can be upper bounded by:
    \begin{equation}
    \label{eq:nizk_ecz}
        \Delta_c(D_S; D_I) \le \Delta_c(D_P; D_I) + \Delta_c(D_S; D_P)  \le \ec(n) + \ez(n).
    \end{equation}
    Combining (\ref{eq:nizk_inv}) and (\ref{eq:nizk_ecz}), since both $f_x$ and $\cA$ are efficient algorithms, we obtain that:
    \begin{align*}
    \label{eq:nizk_yes}
        \Prob{}{R^\cA(x) = 1} &= \Prob{r \gets \gen(1^n)}{f_x(\cA(x;(r,1))) = (r,1)} \nonumber\\
        &= \Prob{(r,a) \gets D_I}{f_x(\cA(x;(r,a))) = (r,a)} \nonumber\\
        &\geq \Prob{(r,a) \gets D_S}{f_x(\cA(x;(r,a))) = (r,a)} \nonumber - \Delta_c(D_S,D_I)\\
        &= \Prob{(r,a) \gets f_x(\cU_\ell)}{f_x(\cA(x;(r,a))) = (r,a)} \nonumber - \Delta_c(D_S,D_I)\\
        &\ge 1 - \frac{1}{p(n)} - (\epsilon_c(n)  + \epsilon_z(n)),
    \end{align*}
    which shows condition \ref{it:nizk_yes}.

    \medskip
    Next, we show that condition \ref{it:nizk_no} holds. When $x\in \zo^n\setminus\cL$, for any polynomial-time algorithm $\cA$, we will show that the soundness guarantees that:
    \begin{equation}
    \label{eq:nizk_no}
        \Prob{}{R^\cA(x) = 1} \le \epsilon_s(n). 
    \end{equation}
    We prove (\ref{eq:nizk_no}) by contradiction. Suppose $\Prob{}{R^\cA(x) = 1} > \es(n)$, which is equivalent to 
    \begin{equation*}
        \Prob{\substack{r \gets \gen(1^n) \\ \rho^* \gets \cA(r,1)}}{f_x(\rho^*) = (r,1)} > \esn. 
    \end{equation*}
    Then, there is a construction of an efficient malicious prover $\sfP^*$ as follows: when receiving $r\gets \gen(1^n)$, it computes $\rho^* \gets \cA(r,1)$, $\pi^* \gets \Sim(x; \rho^*)$ and outputs $\pi^*$.  Note that, when the inverter successfully finds a correct pre-image of $(r,1)$, $\sfP^*$ produces a valid proof $\pi^*$ on which the verifier accepts. Thus,
    \begin{equation*}
        \Prob{\substack{r\gets \gen(1^n) \\ \pi^* \gets \sfP^*(x; r)}}{\sfV(x;r,\pi) = 1} \ge \Prob{\substack{r \gets \gen(1^n) \\ \rho^* \gets \cA(r,1)}}{f_x(\rho^*) = (r,1)}  > \esn, 
    \end{equation*}
    which contradicts the soundness.
\end{proofof}

\begin{corollary}
\label{lem:nizk_aiowf}
    For $\ec, \es, \ez: \Nat \rightarrow [0,1]$, suppose there is an $\NP$ language $\cL \notin \mathsf{ioP/poly}$ that has an $(\ec, \es, \ez)$-NIZK proof or argument with $\ec+\es+\ez <_n 1$. Then auxiliary-input one-way functions exist.
\end{corollary}

\begin{proofof}{\cref{lem:nizk_aiowf}}
    Let $R$ be the oracle-aided algorithm and $\cF$ be the function family guaranteed by Lemma \ref{lem:nizk_red}. 
    There exists a polynomial $q$ such that:
    \begin{equation*}
        \ecn + \esn + \ezn + \frac{1}{q(n)} < 1 
    \end{equation*}
    holds for all sufficiently large $n \in \Nat$. Then, let $p = 2q$. Suppose that there is a (non-uniform) PPT algorithm $\cA$ that, for infinitely many $n\in \Nat$, for every $x\in \cL_n$,  inverts $f_x$ with probability at least $(1- 1/p(n))$. By Lemma \ref{lem:nizk_red}, for every $x\in \cL_n$,
        \begin{equation*}
            \Prob{}{R^{\cA}(x) = 1} \ge 1 - \ecn - \ezn - \frac{1}{p(n)};
        \end{equation*}
    and for every $x\in \zo^n \setminus \cL_n$, 
        \begin{equation*}
            \Prob{}{R^{\cA}(x) = 1} \le \esn.
        \end{equation*}
    Since both $\cA$ and $R$ are polynomial-time algorithms and
    \begin{equation*}
        1 - \ecn - \ezn - \frac{1}{2q(n)} >_n \esn, 
    \end{equation*}
    it contradicts $\cL\notin \mathsf{ioP/poly}$. Hence, $\cF$ is a weak auxiliary-input one-way function: for any efficient non-uniform adversary, for all sufficiently large $n$, there exists a function $f_x$, $x\in \zo^n$ on which the adversary cannot successfully invert $f_x$ with probability at least $(1- 1/p(n))$. By \cite{DBLP:conf/focs/Yao82a}, the existence of weak ai-OWFs implies the existence of (standard) ai-OWFs, which concludes the proof. 
\end{proofof}

\subsection{Reductions from Public-Coin ZK}
\label{sec:pczk_aiowf}

\begin{lemma}
\label{lem:pczk_red}    
    For some $\ec, \es, \ez: \Nat \rightarrow [0,1]$, $t:\Nat \rightarrow \Nat$, suppose a language $\cL\in \NP$ has an $(\ec, \es, \ez)$-public-coin ZK proof or argument with $t$ messages. Then there exists a reduction $R$, which is a polynomial-time oracle-aided algorithm, and a polynomial-time computable function family $\cF = \set{f_x}_{x\in \zo^*}$ such that, for any probabilistic polynomial-time algorithm $\cA$, any polynomial $p$, and all large enough $n\in\Nat$,
    \begin{enumerate}
        \item\label{it:pczk_yes} For any $x\in \cL_n$, if $\cA$ distributionally inverts $f_x$ with deviation at most $1/{p(n)}$, then
        \begin{equation*}
            \Prob{}{R^{\cA}(x) = 1} \ge 1 - \ecn - (t(n)-1)\cdot \ezn - \frac{t(n)}{p(n)}. 
        \end{equation*}
        \item\label{it:pczk_no} For any $x\in \zo^n \setminus \cL_n$, 
        \begin{equation*}
            \Prob{}{R^{\cA}(x) = 1} \le \esn. 
        \end{equation*} 
        \end{enumerate}
\end{lemma}

\begin{proofof}{Lemma \ref{lem:pczk_red}}
    Denote by $n$ the input length. 
    Let $\Sim$ be the simulator that satisfies the zero-knowledge requirement, and suppose it uses $\ell$ randomness bits, where $\ell= \ell(n)$. Let $k=k(n)$ and $t = t(n)$ be the number of rounds and messages, respectively. Let $m_i= m_i(n)$ be the number of bits of public randomness used by the verifier in the $i$-th round, for $i\in [k]$. When $t = 2k$, the verifier sends a random string first in the protocol; when $t = 2k-1$, the prover starts first. For convenience, let $m_1 = 0$ when $t=2k-1$, simply taking the verifier's first message to be an empty string. 
    
    Construct a function family $\cF = \set{f_{x}}_{x\in \zo^*}$ as follows:
    \begin{framed}
        \noindent \underline{$f_{x}(i, \rho)$}: 
        \begin{enumerate}
            \item $(r_1,\pi_1,\dots, r_{k}, \pi_k )\gets \Sim(x; \rho)$
            \item $a \gets \sfV(x; r_1,\pi_1, \dots, r_k,\pi_k)$
            \item If $i=k$ output $(r_1,\pi_1, \dots, r_{i-1}, \pi_{i-1}, r_i, a)$
            \item Else output $(r_1,\pi_1, \dots, r_{i-1}, \pi_{i-1}, r_i, 1)$
        \end{enumerate}
    \end{framed}
    Note that the output length may vary with different inputs. One can equalize the output lengths by padding with a fixed constant string after those short outputs. We keep the current definition for simplicity. Assume that $\cA$ is a PPT algorithm for inverting $f_x$'s, where the output of $\cA$ consists of two parts $(i, \rho)$. Let $\mathcal{A}_2$ denote the algorithm obtained from $\mathcal{A}$ by projecting its output onto the second component $\rho$.
    
    Define an efficient prover $\Tilde{\sfP}$ with oracle access to $\cA$, where we denote by $\Tilde{\sfP}_i$ the prover's algorithm in the $i$-th round. The prover proceeds as follows:
    \begin{framed}
        \noindent \underline{$\Tilde{\sfP}^\cA_i(r_1, \pi_1, \dots, r_i)$}:
        \begin{enumerate}
            \item $\rho \gets \cA_2(x; r_1, \pi_1, \dots, r_i, 1)$
            \item $\pi_i \gets \Sim_{2i}(x; \rho)$
            \item Output $\pi_i$
        \end{enumerate}
    \end{framed}
    The reduction $R^\cA$ runs the following: on input $x$, it simulates the protocol $\prot{\Tilde{\sfP}^\cA}{\sfV}$ (with the same syntax as $\prot{\sfP}{\sfV}$), and outputs $1$ if and only if the protocol $\prot{\Tilde{\sfP}^\cA}{\sfV}(x)$ accepts. Since both prover $\tsfP$ and verifier $\sfV$ are efficient, $R$ runs in polynomial time. 
    We state the following claims.
    \begin{claim}
    \label{claim:pczk_yes}
        When $x\in \cL_n$, if $\cA$ distributionally inverts $f_x$ with deviation at most $1/{p(n)}$
        \begin{equation*}
            \Prob{}{\prot{\Tilde{\sfP}^\cA}{\sfV}(x) = 1 } \ge  1- \ecn - (t(n)-1)\cdot \ezn - \frac{k(n)}{p(n)}.
        \end{equation*}
    \end{claim}
    \begin{claim}
    \label{claim:pczk_no}
        When $x\in \zo^n \setminus \cL_n$ 
        \begin{equation*}
            \Prob{}{\prot{\Tilde{\sfP}^\cA}{\sfV}(x) = 1 } \le \esn. 
        \end{equation*}
    \end{claim}
    Claim~\ref{claim:pczk_no} follows immediately from the soundness of the ZK protocol. We present the proof of Claim~\ref{claim:pczk_yes} in \cref{sec:pczk_yes_proof}. Putting the above claims together completes the proof of the lemma.  
\end{proofof}

\begin{corollary}
\label{lem:pczk_aiowf}
    For $\ec, \es, \ez: \Nat \rightarrow [0,1]$ and $t:\Nat \rightarrow \Nat$, suppose there is an $\NP$ language $\cL \notin \mathsf{ioP/poly}$ that has a $t$-message public-coin $(\ec, \es, \ez)$-ZK proof or argument with $\ec+\es+(t-1)  \ez <_n 1$. Then auxiliary-input one-way functions exist.
\end{corollary}



\begin{proofof}{\cref{lem:pczk_aiowf}}
    Let $R$ be the oracle-aided algorithm and $\cF$ be the function family defined in the proof of Lemma \ref{lem:pczk_red}. Since $\ec+\es+(t-1)  \ez <_n 1$, there exists a polynomial $q$ such that
    \begin{equation*}
        \ecn + \esn + (t(n)-1)\cdot \ezn + \frac{1}{q(n)} < 1 
    \end{equation*}
    holds for all sufficiently large $n \in \Nat$. Then, let $p = 2kq$. Suppose that there is a (non-uniform) PPT algorithm $\cA$ such that, for infinitely many $n\in \Nat$, for every $x\in \cL_n$,  distributional inverts $f_x$ with deviation at most $ 1/p(n)$. By Lemma \ref{lem:pczk_red}, for every $x\in\cL_n$,
        \begin{equation*}
            \Prob{}{R^{\cA}(x) = 1} \ge 1 - \ecn - (t(n)-1)\cdot \ezn - \frac{k(n)}{p(n)};
        \end{equation*}
    and for every $x\in \zo^n \setminus \cL_n$, 
        \begin{equation*}
            \Prob{}{R^{\cA}(x) = 1} \le \esn.
        \end{equation*}
    Both $\cA$ and $R$ are polynomial-time algorithms and
    \begin{equation*}
        1 - \ecn - (t(n)-1)\cdot \ezn - \frac{1}{2q(n)} >_n \esn, 
    \end{equation*}
    which contradicts $\cL\notin \mathsf{ioP/poly}$. Hence, $\cF$ is an auxiliary-input distributional one-way function. By \cite{DBLP:conf/focs/ImpagliazzoL89}, the existence of ai-dOWFs implies the existence of ai-OWFs, which concludes the proof. 
\end{proofof}


\subsubsection{Proof of Claim \ref{claim:pczk_yes}}
\label{sec:pczk_yes_proof}
\begin{proof}
    Let  $x\in \cL_n$, suppose that $\cA$ distributionally inverts $f_x$ with deviation at most $1/{p(n)}$, which means that 
    \begin{equation*}
        \Delta_s\left(\cU, f_x(\cU)  ; \cA(f_x(\cU)), f_x(\cU)  \right) \le \frac{1}{p(n)}.
    \end{equation*}
    We abuse the notation by omitting $x$ from the algorithm's input and the subscript of $\cU$, assume that the length of the uniform distribution $\cU$'s output always matches the input length of $f_{x}$. 

    First, consider the following distribution $D_I$.
    \begin{framed}
    \begin{align*}
        D_I:\quad &r_1 \gets \cU_{m_1}, \pi_1 \gets \tsfP_1^\cA(x; r_1) \\
        & \dots\\
        & r_{k-1} \gets \cU_{m_{k-1}}, \pi_{k-1}\gets \tsfP_{k-1}^\cA(x; r_1,\pi_1, \dots, r_{k-1}) \\
        & r_k \gets \cU_{m_k}\\
        & \text{Output }(r_1, \pi_1,\dots,  r_{k-1}, \pi_{k-1}, r_k, 1)
    \end{align*}  
    \end{framed}
    $D_I$ runs the protocol $\prot{\tsfP^\cA}{\sfV}$ to generate the transcript except for the last proof $\pi_k$, and pads the transcript with $1$.     
    Denote the string by $s = (r_1, \dots, \pi_{k-1}, r_k, a)$. Observe that
    \begin{equation}
    \label{eq:pczk_suc_pr}
        \Prob{}{\prot{\Tilde{\sfP}^\cA}{\sfV}(x) = 1} \ge \Prob{s \gets D_I}{f_x(k,\cA(s)) = s},
    \end{equation}
    since $D_I$ is the distribution of the inputs on which $\tsfP$ invokes $\cA$ before sending their last message $\pi_k$ and whenever the inverter $\cA$ finds a good randomness, $\tsfP_k^\cA$ outputs a valid proof such that the protocol accepts. 
    
    Similarly, define the distributions $D_S$, which samples a transcript from the simulator and outputs the transcript together with the corresponding verification bit.  
    \begin{framed}
    \begin{align*}
        D_S:\quad & \rho \gets \cU_\ell\\
        &(r_1,\pi_1, \dots, r_k, \pi_k) \gets \Sim(x; \rho)\\
        &a \gets \sfV(x; r_1,\pi_1, \dots, r_k,\pi_k)\\
        &\text{output }(r_1,\pi_1,\dots, r_k ,a)
    \end{align*}
    \end{framed}
    Note that the distribution outputting $(k, \rho, s)$ by sampling $s\gets D_S$ and $(k, \rho) \gets \cA(s)$ is exactly the same as the distribution $(\cA(f_x(k, \cU_\ell)), f_x(k, \cU_\ell))$.  
    By the property of indistinguishability, we have
    \begin{align*}
        &\abs{\Prob{(k, \rho, s) \gets (\cA(f_x(k, \cU_\ell)), f_x(k, \cU_\ell))}{f_x(k, \rho) =s} - \Prob{(k, \rho, s) \gets ((k, \cU_\ell), f_x(k, \cU_\ell))}{f_x(k, \rho) = s}}\notag \\ & \quad  \le\Delta_s((k, \cU_\ell), f_x(k, \cU_\ell); \cA(f_x(k, \cU_\ell)), f_x(k, \cU_\ell)), 
    \end{align*}
    where the first probability equals $\Prob{s \gets D_S}{f_x(\cA(s)) = s}$ and the second is exactly $1$.
    Thus, we derive
    \begin{equation}
    \label{eq:pczk_inv_pr}
        \Prob{s \gets D_S}{f_x(\cA(s)) = s} \ge 1 - \Delta_s((k, \cU_\ell), f_x(k, \cU_\ell); \cA(f_x(k, \cU_\ell)), f_x(k, \cU_\ell)). 
    \end{equation}
    Towards proving an upper-bound on $\Delta(D_S, D_I)$ to connect (\ref{eq:pczk_suc_pr}) and (\ref{eq:pczk_inv_pr}), we define the following distributions as intermediate hybrids, for $i\in [k]$. Below, $w$ is any valid witness for $x$ that satisfies $(x,w)\in\cR_\cL$.
    \begin{framed}
    \begin{align*}
        D_S^{(i)}:\quad &\rho \gets \cU_{\ell}\\
        &(r_1,\pi_1, \dots, r_{i-1}, \pi_{i-1}) \gets \Sim_{1\dots 2(i-1)}(x; \rho)\\
        &r_i \gets \cU_{m_i}, \pi_i \gets \Tilde{\sfP}^\cA_i(x; r_1,\dots, \pi_{i-1}, r_i) \\
        &\dots\\
        & r_{k-1} \gets \cU_{m_{k-1}}, \pi_{k-1} \gets  \tsfP^\cA_{k-1}(x; r_1,\dots, \pi_{i-1}, r_i, \dots, r_{k-1})\\
        & r_k \gets \cU_{m_k}\\
        & \text{output }(r_1, \dots, \pi_{i-1}, r_i, \dots, r_k, 1)
    \end{align*}
    \end{framed}


    \begin{framed}
    \begin{align*}
        D_P^{(i)}:\quad &r_1\gets \cU_{m_1}, \pi_1 \gets \sfP_1(x, w; r_1) \\
        & \dots\\
        & r_{i-1} \gets \cU_{m_{i-1}}, \pi_{i-1} \gets \sfP_{i-1}(x, w; r_1, \dots, r_{i-1})\\
        & r_{i} \gets \cU_{m_{i}}, \pi_i \gets \tsfP^\cA_i(x; r_1, \dots, \pi_{i-1}, r_i)\\
        & \dots\\
        & r_{k-1} \gets \cU_{m_{k-1}}, \pi_{k-1} \gets  \tsfP^\cA_{k-1}(x; r_1,\dots, \pi_{i-1}, r_i, \dots, r_{k-1})\\
        & r_k \gets \cU_{m_k}\\
        & \text{output }(r_1,\dots, \pi_{i-1}, r_i, \dots,r_k, 1)
    \end{align*}
    \end{framed}

    The messages in the first $(i-1)$ rounds are sampled by the simulator $\Sim$ in $D_S^{(i)}$ while those in $D_P^{(i)}$ are sampled by the protocol $\prot{\sfP}{\sfV}$ run by the honest prover with an $\NP$ witness $w$.
    
    In both distributions, starting from the $i$-th iteration, the remaining transcripts are sampled following the protocol with $\Tilde{\sfP}$. 
    Clearly, $D_S^{(1)} = D_P^{(1)} = D_I$. 
    By data processing inequality, we have 
    \begin{equation}
    \label{eq:pczk_bd_si_pi}
        \Delta_c\left(D_S^{(i)}; D_P^{(i)}\right) < \Delta_c(\Sim(x); \tr(x,w)) < \ezn. 
    \end{equation}
    
    Then, we connect $D_P^{(i)}$ and $D_S^{(i+1)}$ by modifying the sampling strategy in $D_S^{(i)}$ slightly. For $i\in [k]$, define $D_M^{(i)}$,
    \begin{framed}
    \begin{align*}
        D_M^{(i)}:\quad &\rho \gets \cU_{\ell}\\
        &(r_1,\pi_1, \dots, \pi_{i-1}, {\color{red} r_i}) \gets \Sim_{1\dots (2i-1)}(x; \rho)\\
        & {\color{red} \pi_i \gets \tsfP^\cA_i(x; r_1, \dots, \pi_{i-1}, r_i)}\\
        & r_{i+1} \gets \cU_{m_{i+1}}, \pi_{i+1} \gets \tsfP^\cA_{i+1}(x; r_1, \dots, \pi_{i-1}, r_i, \pi_i, r_{i+1})\\
        & \dots\\
        & r_{k-1} \gets \cU_{m_{k-1}}, \pi_{k-1} \gets  \tsfP^\cA_{k-1}(x; r_1,\dots, \pi_{i-1}, r_i, \dots, r_{k-1})\\
        & r_k \gets \cU_{m_k}\\
        & \text{output }(r_1,\dots, \pi_{i-1}, r_i, \pi_i \dots,r_k, 1)
    \end{align*}
    \end{framed}
    Observe that $(r_1, \dots, \pi_{i-1}, r_i)$ is sampled from the simulator in $D_M^{(i)}$ while it is sampled by the protocol $\prot{\sfP}{\sfV}$ in $D_{P}^{(i)}$; and the later part of the outputs $(\pi_i, \dots, r_k)$ are generated in the same way in both $D_M^{(i)}$ and $D_P^{(i)}$. By the data processing inequality
    \begin{equation}
    \label{eq:pczk_bd_mi_pi}
        \Delta_c\left(D_M^{(i)}, D_P^{(i)}\right) \le \Delta_c(\Sim(x); \tr(x,w)) \le \ezn. 
    \end{equation}
    Note that when the first message is sent by the prover, that is $t = 2k-1$, then the distributions $D_M^{(1)}$ and $D_P^{(1)}$ are identical. 
    
    The only difference between $D_S^{(i+1)}$ and $D_M^{(i)}$ is how the $i$-th proof is sampled. 
    For $D_S^{(i+1)}$, equivalently, $(r_1, \pi_1, \dots, r_i, \pi_i)$ is generated by: 
    \begin{align*}
        &\rho \gets \cU_\ell \\
        &(r_1, \pi_1, \dots, r_i, 1) \gets f_x(i, \rho)\\
        &\pi_i \gets \Sim_{2i}(x; \rho)
    \end{align*}
    Recall the prover's strategy $\tsfP$, $(r_1, \pi_1, \dots, r_i, \pi_i)$ in $D_M^{(i)}$ is sampled by  
    \begin{align*}
        &\rho \gets \cU_\ell\\
        &(r_1, \pi_1, \dots, r_i, 1) \gets f_x(i, \rho)\\
        &\hat{\rho} \gets \cA_2(r_1, \pi_1, \dots, r_i, 1)\\
        &\pi_i \gets \Sim_{2i}(x; \hat{\rho})
    \end{align*}
    The distance between their marginal distributions on  $(r_1, \pi_1, \dots, r_i, \pi_i)$ is upper bounded by $\Delta_s((i, \cU_\ell), f_x(i, \cU_\ell); \cA(f(i, \cU_\ell)), f_x(i, \cU_\ell))$.
    The remaining outputs of $D_S^{(i+1)}$ and $D_M^{(i)}$ follow the same sampling procedures, based on the previous transcript $(r_1, \pi_1, \dots, r_i, \pi_i)$. Then by the data processing inequality, we obtain
    \begin{equation}
    \label{eq:pczk_bd_si_mi}
        \Delta_s\left(D_S^{(i+1)}, D_M^{(i)}\right) \le \Delta_s\left((i, \cU_\ell), f_x(i, \cU_\ell); \cA(f(i, \cU_\ell)), f_x(i, \cU_\ell)\right), 
    \end{equation}
    
    Consider the distance between $D_P^{(k)}$ and the following distribution $D_P$:
    \begin{framed}
    \begin{align*}
        D_P:\quad &r_1\gets \cU_{m_1}, \pi_1 \gets \sfP_1(x, w; r_1) \\
        & \dots\\
        & r_k \gets \cU_{m_k}, \pi_{k} \gets \sfP_k(x, w; r_1, \dots, r_{k})\\
        & a \gets \sfV(x; r_1, \dots, \pi_k)\\
        & \text{output }(r_1,\dots, \dots,r_k, a)
    \end{align*}
   \end{framed}
   The only difference between $D_P$ and $D_P^{(k)}$ is that, the last verification bit in $D_P$ is sampled by the verifier $\sfV$ while the counterpart in $D_P^{(k)}$ is assigned to be $1$. 
    By completeness and using the same argument as (\ref{eq:nizk_ec})
    \begin{equation}
    \label{eq:pczk_bd_p_pk}
        \Delta_s\left(D_P; D_P^{(k)}\right) \le \ecn. 
    \end{equation}
    
    By zero-knowledge, we have 
    \begin{equation}
    \label{eq:pczk_bd_s_p}
        \Delta_c(D_S; D_P) \le \Delta_c(\Sim(x); \tr(x,w)) \le \ezn. 
    \end{equation}

    Combining (\ref{eq:pczk_bd_si_pi}), (\ref{eq:pczk_bd_mi_pi}), (\ref{eq:pczk_bd_si_mi}), (\ref{eq:pczk_bd_p_pk}) and (\ref{eq:pczk_bd_s_p}), it follows from the triangle inequality that, when $t = 2k$ 
    \begin{align*}
        \Delta_c(D_S;  D_I) &\le \Delta_c(D_S; D_P) + \Delta_c\left(D_P; D_P^{(k)}\right)\notag\\ &+ \sum_{i=1}^{k-1} \left( \Delta_c\left(D_P^{(i+1)}; D_S^{(i+1)}\right) + \Delta_c\left(D_S^{(i+1)}; D_M^{(i)}\right) + \Delta_c\left(D_M^{(i)}; D_P^{(i)}\right)\right) \notag \\
        & \le \ecn + (2k-1)\cdot \ezn + \sum_{i\in [k-1]} \Delta_s((i, \cU_\ell), f_x(i, \cU_\ell); \cA(f(i, \cU_\ell)), f_x(i, \cU_\ell)) \\
        & \le \ecn + (2k-1)\cdot \ezn + \frac{k}{p(n)} .
    \end{align*}
    When $t = 2k-1$, $\Delta_c(D_M^{(1)}, D_P^{(1)}) = 0$, thus
    \begin{equation*}
        \Delta_c(D_S;  D_I) \le \ecn + (2k-2)\cdot \ezn + \frac{k}{p(n)} .
    \end{equation*}
    Note that the last inequality holds because, assuming without loss of generality that the first part of the output of $\cA$ is always correct, 
    \begin{equation*}
        \frac{1}{k}\sum_{i \in [k]} \Delta_s((i, \cU_\ell), f_x(i, \cU_\ell); \cA(f(i, \cU_\ell)), f_x(i, \cU_\ell))  = \Delta_s\left( \cU, f_x(\cU)  ; \cA(f_x(\cU)), f_x(\cU)  \right).
    \end{equation*}
    We conclude that 
    \begin{align*}
        \Prob{}{\prot{\Tilde{\sfP}^\cA}{\sfV}(x) = 1} &\ge \Prob{s \gets D_I}{f_x(\cA(s)) = s}\\
        &\ge \Prob{s \gets D_S}{f_x(\cA(s)) = s} - \Delta_c(D_S, D_I)\\
        &\ge 1 - \ecn - (t-1) \cdot \ezn - \frac{k}{p(n)}. 
    \end{align*}
    This proves the claim.
\end{proof}

\subsection{Reductions from Constant-Round Public-Coin ZK}
\label{sec:constant_aiowf}

For the lemma below, as opposed to the other (similar) ones, we consider the number of rounds in the protocol rather than the number of messages. Recall that a round consists of a message from the verifier, followed by a message from the prover. If the protocol starts with the prover sending a message, we think of it as starting with an empty message from the verifier instead.

\begin{lemma}
\label{lem:cr_red}
    For some $\ec, \es, \ez: \Nat \rightarrow [0,1]$ and any constant $k\in \Nat$, suppose a language $\cL\in \NP$ has a $k$-round public-coin $(\ec, \es, \ez)$-ZK proof or argument. Then there exists a polynomial-time oracle-aided algorithm $R$ and families of oracle-aided functions $\cF_1, \dots, \cF_k$ satisfying the following. 
    
    For $i\in[k]$, the family $\cF_{i} = \set{f_{i,x}}_{x\in\zo^*}$ consists of functions that require access to $(k-i)$ oracles, and are polynomial-time computable given these oracles. For any sequence of probabilistic polynomial-time algorithms $\cA_1,\dots, \cA_k$, any polynomial $p$, and all large enough $n$, we have the following.
    \begin{enumerate}
        \item For any $x\in \cL_n$, if for all $i\in [k]$, $\cA_i$ distributionally inverts $f_{i,x}^{\cA_{i+1},\dots,\cA_k}$ with deviation at most $1/{p(n)}$, then
        \begin{equation*}
            \Prob{}{R^{\cA_1\dots \cA_k}(x) = 1} \ge 1 - \ecn - k\cdot \ezn - \frac{3k-2}{p(n)} - \negl(n). 
        \end{equation*}
        \item For any $x\in \zo^n \setminus \cL_n$, 
        \begin{equation*}
            \Prob{}{R^{\cA_1\dots \cA_k}(x) = 1} \le \esn. 
        \end{equation*}
    \end{enumerate}
\end{lemma}

\begin{proofof}{Lemma \ref{lem:cr_red}}
    Denote by $n$ the input length, and fix any polynomial $p$. 
    Let $\Sim$ be the simulator that satisfies the zero-knowledge requirement, and suppose it uses $\ell$ randomness bits, where $\ell= \ell(n)$. Let $k$ be the number of rounds which is a constant, and $m_i= m_i(n)$ be the number of bits of public randomness used by the verifier in the $i$-th round, for $i\in [k]$.

    \medskip
    Construct a function family $\cF_k = \set{f_{k,x}: x\in \zo^n}_{n \in \Nat}$ as follows 
    \begin{framed}
        \noindent \underline{$f_{k,x}(\rho)$}:
        \begin{enumerate}
            \item $(r_1, \pi_1, \dots, r_k, \pi_k) \gets \Sim(x;\rho)$
            \item $a \gets \sfV(r_1, \pi_1, \dots, r_k, \pi_k)$
            \item Output $(r_1, \pi_1, \dots, r_k, a)$
        \end{enumerate}
    \end{framed}
    Let $\cA_k$ be a polynomial-time algorithm for inverting $\cF_k$. Note that these $\cA_i$'s (defined later for $i \in [k-1]$) are potentially randomized algorithms. To leverage it in the construction of deterministic functions and oracles, we make their internal randomness explicit by treating the randomness as part of the input. Specifically, we write the algorithms as $\cA_{i}(x;r_1, \pi_1, \dots, r_i, a; \rand)$, where $x$ is the auxiliary input indicating that $\cA_i$ is supposed to invert the function $f_{i,x}^{\cA_{i+1},\dots,\cA_k}$ on output $(r_1,\dots,r_i,a)$, and $\rand$ is the randomness that $\cA_i$ uses. The same applies to $\cB_i$'s (also defined later). 
    
    We next define an oracle-aided algorithm $\cB_k$ that, given access to $\cA_k$, essentially measures the acceptance likelihood of a given transcript and outputs a proof $\pi_k$ as follows. 
    \begin{framed}
        \noindent \underline{$\cB_{k}^{\cA_k}(x; r_1,\pi_1, \dots, r_k)$}:
        \begin{enumerate}
            \item $\rho \gets \cA_{k}(x; r_1, \pi_1, \dots, r_k, 1)$
            \item If $f_{k,x}(\rho) = (r_1, \pi_1, \dots, r_k, 1)$
            \begin{itemize}
                \item $\pi_k \gets \Sim_{2k}(x; \rho)$
                \item Output $(1, \pi_k)$
            \end{itemize}
            \item Else output $(0, \bot)$
        \end{enumerate}
    \end{framed}
    
    When algorithm $\cA_k$ successfully finds a pre-image $\rho$, such that the corresponding transcript is accepted by the verifier, then $\cB_k$ outputs $1$ and a valid proof $\pi_k$ such that the verification accepts the entire transcript $(r_1, \pi_1,\dots, r_k, \pi_k)$. Otherwise, $\cB_k$ outputs $(0,\bot)$ which represents failure of inversion. 

    For $i\in [k]$, we write $\cB_{i,1}$ and $\cB_{i,2}$ to denote the algorithm that outputs only the first component and the second of $\cB_i$'s output respectively, and let $\cB_i = \cB_i^{\cA_i\dots \cA_k}$ for simplicity. Later, $\cB_i$'s are defined with respect to $\cA_i$'s for $i\in [k-1]$. 

    \medskip
    For $i\in[k-1]$, define the family $\cF_i = \set{f_{i,x}}_{x\in \zo^*}$ with oracle access to $\cA_{i+1},\dots,\cA_k$ as follows. Here, we set $q = n \cdot p(n)^2$. The input of function $f_{i,x}$ consists of a string $\rho$ of length suitable for the randomness of $\Sim$, strings $\sigma_{1}, \dots, \sigma_{q}$ each of length $m_{i+1}$ corresponding to the $(i+1)$-th verifier's message (public coins), and strings $\rand_1, \dots, \rand_q$ each of the length suitable for the randomness of $\cB_{i+1}$.
    
    \begin{framed}
        \noindent \underline{$f^{\cA_{i+1} \dots \cA_k}_{i,x}\left(\rho, \sigma_{1}, \dots, \sigma_{q}, \rand_1, \dots, \rand_q \right)$}:
        \begin{enumerate}
            \item $(r_1, \pi_1, \dots, r_k, \pi_k) \gets \Sim(x;\rho)$
            \item $\est \gets \frac{1}{q}\sum_{j\in [q]} \cB_{i+1, 1}^{\cA_{i+1}\dots \cA_k}(x; r_1, \pi_1, \dots, r_i, \pi_i, \sigma_{j}; \rand_j)$
            \item Output $(r_1, \pi_1, \dots, r_i, \est)$
        \end{enumerate}
    \end{framed}
    \noindent With an overwhelming probability, the value $\est$ reflects the value of its expectation 
    \begin{equation*}
        \Expec{}{\est} = \Expec{r_{i+1} \gets \cU_{m_{i+1}}}{\cB_{i+1,1}(r_1, \pi_1, \dots, r_i,\pi_i, r_{i+1})}, 
    \end{equation*}
    which in turn measures the ability of the inversion by $\cA_{i+1}$ to find a completion of the simulated partial transcript for the first $i$ rounds specified by $\rho$ to an accepting one.
    
    Similarly, let $\cA_i$ be a polynomial-time algorithm to potentially invert $\cF_i$. 
    Let $ \tau =  1/p(n)$. 
    An efficient oracle-aided algorithm $\cB_i$ is defined as follows.
    \begin{framed}
        \noindent \underline{$\cB_{i}^{\cA_i\dots \cA_k}(x; r_1,\pi_1, \dots, r_i)$}:
        \begin{enumerate}
            \item For $a = q^{k-i}, q^{k-i}-1, \dots, 2, 1$ in decreasing order
            \begin{itemize}
                \item $(\rho, \sigma_{1}, \dots, \sigma_{q}, \rand_1, \dots, \rand_q) \gets \cA_{i}(x; r_1, \pi_1, \dots, r_i, a/q^{k-i})$
                \item If $f_{i,x}^{\cA_{i+1},\dots,\cA_k}(\rho, \sigma_{1}, \dots, \sigma_{q}, \rand_1, \dots, \rand_q) = (r_1, \pi_1, \dots, r_i, a/q^{k-i})$
                \begin{itemize}
                    \item $\pi_i \gets \Sim_{2i}(x; \rho)$
                    \item $\hat{\sigma}_{1}, \dots, \hat{\sigma}_{q} \gets \cU_{m_{i+1}}$
                    \item $\hat{\est} \gets \frac{1}{q}\sum_{j\in [q]} \cB_{i+1, 1}^{\cA_{i+1}\dots \cA_k}(x; r_1, \pi_1, \dots, r_i, \pi_i, \hat{\sigma}_{j})$
                    \item If $\abs{\hat{\est} - a/q^{k-i}}< \tau$
                    \begin{itemize}
                        \item Output $(a/q^{k-i}, \pi_i)$
                    \end{itemize}
                \end{itemize}
            \end{itemize}
            \item Output $(0,\bot)$
        \end{enumerate}
    \end{framed}
    The output consists of two parts: a value $a/q^{k-1}$ in $[0,1]$ and a proof $\pi_i$. With an overwhelming probability, the value $a/q^{k-1}$ estimates well the acceptance probability when the prover's response to $(r_1,\pi_1,\dots,r_i)$ is $\pi_i$. This is formally defined and proved later. The algorithm iterates over $a$ from $q^{k-i}$ to $1$ to find an almost optimal proof that maximizes the acceptance rate, which is utilized to construct the distinguisher later. 
    
    Equipped with the algorithms $\cB_i$'s, we are ready to construct a polynomial-time prover $\tsfP$ with oracle access to $\cA_i$'s. Denote by $\tsfP_i$ the prover's algorithm in the $i$-th round. 
    \begin{align*}
        \tsfP_i^{\cA_i\dots \cA_k}(x; r_1, \pi_1, \dots, r_i) = \cB^{\cA_i\dots \cA_k}_{i,2}(x; r_1, \pi_1, \dots, r_i)
    \end{align*}
    The reduction $R$ runs the protocol $\prot{\tsfP}{\sfV}(x)$ itself and outputs $1$ if and only if the protocol accepts. The efficiency of $\tsfP$ ensures the polynomial runtime of $R$. We state the following claims. 
    \begin{claim}
    \label{claim:crpc_yes}
        When $x\in \cL_n$, if for all $i\in [k]$, $\cA_i$ distributionally inverts $f_{i,x}^{\cA_{i+1},\dots,\cA_k}$ with deviation at most $1/{p(n)}$, then
        \begin{equation*}
            \Prob{}{\prot{\Tilde{\sfP}^{\cA_1\dots\cA_k}}{\sfV}(x) = 1 } \ge 1- \ecn - k\cdot \ezn -  \frac{3k-2}{p(n)} - \negl(n).
        \end{equation*}
    \end{claim}
    \begin{claim}
    \label{claim:crpc_no}
        When $x\in \zo^n\setminus \cL_n$, 
        \begin{equation*}
            \Prob{}{\prot{\Tilde{\sfP}^{\cA_1\dots\cA_k}}{\sfV}(x) = 1 } \le \esn. 
        \end{equation*}
    \end{claim}
    Claim~\ref{claim:crpc_no} follows immediately from the soundness of the ZK protocol. We present the proof of Claim~\ref{claim:crpc_yes} in \cref{sec:crpc_yes_proof}. Putting the above claims together completes the proof of the lemma.  
    
\end{proofof}

\begin{remark}
    In this approach, the functions are defined recursively and the number of function families is closely related to the round number $k$, so we can only achieve the implication result for any constant $k$. We do not know how to make it work beyond a finite number of rounds. 
    This issue has also been observed in other work that employ similar recursive constructions~\cite{DBLP:conf/crypto/BuzekT24, DBLP:conf/tcc/MuV24}.
\end{remark}

\begin{corollary}
\label{lem:cr_aiowf}    
    For $\ec, \es, \ez: \Nat \rightarrow [0,1]$ and a constant $k\in\Nat$, suppose there is an $\NP$ language $\cL\notin \mathsf{P/poly}$ that has a $k$-round public-coin $(\ec, \es, \ez)$-ZK proof or argument with $\ec+\es+ k \cdot \ez <_n 1$. Then infinitely-often auxiliary-input one-way functions exist.
\end{corollary}

\begin{proofof}{\cref{lem:cr_aiowf}}
    Let $R$ be the oracle-aided algorithm and $\cF_1,\dots, \cF_k$ be the oracle-aided function family defined in the proof of Lemma \ref{lem:cr_red}. Since $\ec+\es+k \cdot \ez <_n 1$, there exists a polynomial $q$ such that
    \begin{equation*}
        \ecn + \esn + k\cdot \ezn + \frac{1}{q(n)} < 1 
    \end{equation*}
    holds for all sufficiently large $n \in \Nat$. Then, let $p = 3kq$. Suppose that there is no infinitely-often ai-dOWF, which means that there are (non-uniform) PPT algorithms $\cA_1,\dots, \cA_k$ such that, for sufficiently large $n\in \Nat$, for every $x\in \cL_n$ and all $i\in [k]$, $\cA_i$ distributionally inverts $f_{i, x}$ with deviation at most $ 1/p(n)$. By Lemma \ref{lem:pczk_red}
        \begin{equation*}
            \Prob{}{R^{\cA_1\dots\cA_k}(x) = 1} \ge 1 - \ecn - k\cdot \ezn - \frac{3k-2}{p(n)} - \negl(n);
        \end{equation*}
    and for every $x\in \zo^n \setminus \cL_n$, 
        \begin{equation*}
            \Prob{}{R^{\cA_1\dots\cA_k}(x) = 1} \le \esn.
        \end{equation*}
    But the $\cA_i$'s and $R$ are polynomial-time algorithms and
    \begin{equation*}
        1 - \ecn -k \cdot \ezn - \frac{3k-2}{3k}\cdot \frac{1}{q(n)} >_n \esn, 
    \end{equation*}
    which contradicts $\cL\notin \mathsf{P/poly}$. Hence, $\cF$ is an auxiliary-input distributional one-way function. By \cite{DBLP:conf/focs/ImpagliazzoL89}, the existence of ai-dOWFs implies the existence of ai-OWFs, which concludes the proof. 
\end{proofof}

\begin{remark}\label{rem:ioOwf}
    Our approach to obtaining a contradiction relies on the overlap of input lengths on which all of the algorithms $\cA_1, \dots, \cA_k$ succeed in their respective inversions. The negation of this assumption only implies the existence of infinitely-often one-way functions. Given that we only get infinitely-often security anyway, we are able to weaken our assumption to $\cL\notin\mathsf{P/poly}$ instead of needing $\cL\notin\mathsf{ioP/poly}$ as in our other results.
\end{remark}


\subsubsection{Proof of Claim \ref{claim:crpc_yes}}
\label{sec:crpc_yes_proof}


Before proceeding with Claim \ref{claim:crpc_yes}, first consider the value
\begin{equation}
\label{eq:crpc_est_acc}
    \cB(x) = \Expec{r_1 \gets \cU_{m_1}}{\cB_{1,1}(x; r_1)}.
\end{equation}
For simplicity, we omit the superscript that denotes the oracle access throughout this section. 
We note that $\cB$ approximates the acceptance probability of $\prot{\tsfP}{\sfV}$ well, that is
\begin{equation*}
    \cB(x) \approx \Prob{}{\prot{\tsfP}{\sfV}(x) = 1},
\end{equation*}
which is formally stated as follows. 
\ynote{Independent of the success probability of inverting}
\begin{claim}
\label{claim:crpc_est_prot}
    For $k,p$ as defined in Lemma \ref{lem:cr_red}, where $k$ is a constant representing the number of rounds in the protocol and $p$ is a polynomial, 
    we have 
    \begin{equation*}
        \abs{\Prob{}{\prot{\tsfP}{\sfV}(x) = 1} - \cB(x) } \le \frac{2(k-1)}{p(n)} + \negl(n).
    \end{equation*}
\end{claim}

Claim \ref{claim:crpc_est_prot} states that $\cB$ provides a good approximation of the probability that $\prot{\tsfP}{\sfV}$ accepts on $x$. To derive the desired lower bound on this probability as stated in Claim \ref{claim:crpc_yes}, for $x\in \cL_n$, we therefore establish a lower bound on $\cB$ assuming that all $\cA_i$'s distributionally inverts $f_{i,x}$ well. 

\begin{claim}
\label{claim:crpc_est_bd}
    For $x\in \cL_n$ and $\ec, \ez, k, p$ as defined in Lemma \ref{lem:cr_red}, where $k$ is a constant representing the number of rounds in protocol and $p$ is a polynomial, we assume that for all $i\in [k]$, the efficient algorithm $\cA_{i}$ satisfies that 
    \begin{equation}
    \label{eq:crpc_inv_k}
        \Delta_s(\cU, f_{i,x}(\cU); \cA_{i}(f_{i,x}(\cU)), f_{i,x}(\cU)) < \frac{1}{p(n)}
    \end{equation}
    then we have 
    \begin{equation*}
        \cB(x) \ge 1 - \ecn - k \cdot \ezn - \frac{k}{p(n)} - \negl(n). 
    \end{equation*}
\end{claim}

With Claim \ref{claim:crpc_est_prot} and Claim \ref{claim:crpc_est_bd}, both of whose proofs are provided immediately afterwards, we are ready to prove Claim \ref{claim:crpc_yes}.

\begin{proofof}{Claim \ref{claim:crpc_yes}}
    From Claim \ref{claim:crpc_est_prot}, we have
    \begin{equation*}
         \abs{\Prob{}{\prot{\tsfP}{\sfV}(x) = 1} - \cB(x) }\le \frac{2(k-1)}{p(n)} + \negl(n). 
    \end{equation*}
    For $x\in \cL_n$, by Claim \ref{claim:crpc_est_bd}, $\cB(x)$ can be lower bounded by
    \begin{equation*}
        \cB(x) \ge 1 - \ecn - k \cdot \ezn - \frac{k}{p(n)} - \negl(n). 
    \end{equation*}
    Therefore, 
    \begin{equation*}
        \Prob{}{\prot{\tsfP}{\sfV}(x) = 1} \ge 1- \ecn - k\cdot \ezn - \frac{3k-2}{p(n)} - \negl(n). 
    \end{equation*}
\end{proofof}

\begin{proofof}{Claim \ref{claim:crpc_est_prot}}
    For $i\in [k]$, denote by $\cT_{i}(x; r_1, \pi_1, \dots, r_{i})$ the following value, which represents the true probability that $\prot{\tsfP}{\sfV}(x)$ accepts for the first $(2i-1)$ messages being $(r_1, \pi_1, \dots, r_{i})$ in the protocol. In particular, 
    \begin{align*}
        \cT_{i}(x; r_1, \pi_1, \dots, r_i)&=\Prob{\substack{\pi_i \gets \tsfP_i(x; r_1,\pi_1, \dots, r_i)\\\cdots\\ r_k \gets \cU_{m_k}, \pi_k \gets \tsfP_k(x; r_1,\pi_1, \dots, r_k)}}{\sfV(x; r_1,\dots, \pi_k) = 1}. 
    \end{align*}
    Observe that $\cT_{i}(x; r_1, \pi_1, \dots, r_i)$ can be defined recursively from $\cT_{i+1}(x; r_1, \pi_1, \dots, r_{i+1})$
    \begin{align}
    \label{eq:cr_true_rec}
        \cT_{i}(x; r_1, \pi_1, \dots, r_i)&=\Prob{\substack{\pi_i \gets \tsfP_i(x; r_1,\pi_1, \dots, r_i)\\\cdots\\ r_k \gets \cU_{m_k}, \pi_k \gets \tsfP_k(x; r_1,\pi_1, \dots, r_k)}}{\sfV(x; r_1,\dots, \pi_k) = 1}\notag\\
        &=\Prob{\substack{\pi_i \gets \tsfP_i(x; r_1,\pi_1, \dots, r_i)\\r_{i+1} \gets \cU_{m_{i+1}}, \pi_{i+1} \gets \tsfP_i(x; r_1,\pi_1, \dots, r_{i+1})\\ \cdots\\ r_k \gets \cU_{m_k}, \pi_k \gets \tsfP_k(x; r_1,\pi_1, \dots, r_k)}}{\sfV(x; r_1,\dots, \pi_k) = 1}\notag\\
        & = \Expec{\substack{\pi_i \gets \cB_{i,2}(x; r_1,\pi_1, \dots, r_i)\\ r_{i+1} \gets \cU_{m_{i+1}}}}{\cT_{i+1}(x;r_1, \pi_1, \dots, r_i, \pi_i, r_{i+1})}.
    \end{align}
    For the last equality, we remind the reader that $$\tsfP_i(x; r_1, \pi_1, \dots, r_i) = \cB_{i,2}(x; r_1,\pi_1, \dots, r_i).$$ We also note that the probability that $\prot{\tsfP}{\sfV}(x)$ accepts can be written as
    \begin{equation}
    \label{eq:crpc_real_acc}
        \Prob{}{\prot{\tsfP}{\sfV}(x) = 1} = \Expec{r_1 \gets \cU_{m_1}}{\cT_{1}(x;r_1)}. 
    \end{equation}
    
    In the following, we show a more general statement: for any $i\in [k]$, for any input $(x; r_1,\pi_1, \dots, r_i)$, the expected value of the output $\cB_{i,1}$ is close to $\cT_{i}$ 
    \begin{align}
    \label{eq:crpc_est_prot_i}
        \abs{\cT_{i}(x; r_1,\pi_1, \dots, r_i) - \Expec{}{\cB_{i,1}(x; r_1,\pi_1, \dots, r_i)}} \le 2(k-i)(\tau + e^{-2\tau^2q})
    \end{align}
    where $\tau$ and $q$ are the parameters defined in $\cB_i$'s to be $1/p(n)$ and $n\cdot p(n)^2$, respectively. 
    This inequality implies that for any $i\in [k]$, the expected value of $\cB_{i,1}(x; r_1, \pi_1,\dots, r_i)$ always reflects the probability that the protocol $\prot{\tsfP}{\sfV}(x)$ accepts conditioned on the first $(2i-1)$ messages being $(r_1, \pi_1,\dots, r_i)$.

    We prove (\ref{eq:crpc_est_prot_i}) by induction. When $i=k$, on input $(x; r_1,\pi_1, \dots, r_k)$, recall that the algorithm $\cB_{k}$ proceeds as follows. 
    \begin{framed}
        \noindent \underline{$\cB_{k}(x; r_1,\pi_1, \dots, r_k)$}:
        \begin{enumerate}
            \item $\rho \gets \cA_{k}(r_1, \pi_1, \dots, r_k, 1)$
            \item If $f_{k,x}(\rho) = (r_1, \pi_1, \dots, r_k, 1)$
            \begin{itemize}
                \item $\pi_k \gets \Sim_{2k}(x; \rho)$
                \item Output $(1, \pi_k)$
            \end{itemize}
            \item Else output $(0, \bot)$
        \end{enumerate}
    \end{framed}
    The probability that $\sfV$ accepts $(x; r_1,\pi_1, \dots, r_k)$ when the last proof $\pi_k$ is generated by $\tsfP_k = \cB_{k,2}$ is exactly the probability that $\cB_{k, 1}$ outputs $1$, since once $f_{k,x}(\rho) = (r_1, \pi_1, \dots, r_k, 1)$ holds, the $\pi_k$ corresponding to $\rho$ satisfies $\sfV(x; r_1, \pi_1, \dots, r_k, \pi_k)=1$. 
    \begin{align*}
        \cT_{k}(x; r_1,\pi_1, \dots, r_k) 
        &= \Prob{\pi_k\gets \cB_{k,2}(x; r_1, \pi_1, \dots, r_k)}{\sfV(x; r_1, \pi_1, \dots, r_k, \pi_k) = 1}\\ 
        &= \Prob{}{\cB_{k,1}(x; r_1,\pi_1, \dots, r_k) = 1}\\
        &= \Expec{}{\cB_{k,1}(x; r_1,\pi_1, \dots, r_k)}.
    \end{align*}
    Thus, (\ref{eq:crpc_est_prot_i}) is satisfied when $i=k$. 

    Suppose that $i = l + 1$ meets the condition (\ref{eq:crpc_est_prot_i}) for $l\in [k-1]$. 
    Recall the adaptive construction of $\cB_{l}$'s, 
    \begin{framed}
        \noindent \underline{$\cB_{l}(x; r_1,\pi_1, \dots, r_l)$}:
        \begin{enumerate}
            \item For $a = q^{k-l}, q^{k-l}-1, \dots, 2, 1$ in decreasing order
            \begin{itemize}
                \item $(\rho, \sigma_{1}, \dots, \sigma_{q}, \rand_1, \dots, \rand_q) \gets \cA_{l}(r_1, \pi_1, \dots, r_l, a/q^{k-l})$
                \item If $f_{l,x}(\rho, \sigma_{1}, \dots, \sigma_{q}, \rand_1, \dots, \rand_q) = (r_1, \pi_1, \dots, r_l, a/q^{k-l})$
                \begin{itemize}
                    \item $\pi_l \gets \Sim_{2l}(x; \rho)$
                    \item $\hat{\sigma}_{1}, \dots, \hat{\sigma}_{q} \gets \cU_{m_{l+1}}$
                    \item $\hat{\est} \gets \frac{1}{q}\sum_{j\in [q]} \cB_{l+1, 1}(x; r_1, \pi_1, \dots, r_l, \pi_l, \hat{\sigma}_{j})$
                    \item If $\abs{\hat{\est} - a/q^{k-l}}< \tau$
                    \begin{itemize}
                        \item Output $(a/q^{k-l}, \pi_l)$
                    \end{itemize}
                \end{itemize}
            \end{itemize}
            \item Output $(0,\bot)$
        \end{enumerate}
    \end{framed}
    Denote the output of $\cB_l$ by $(\est_o, \pi_l)$. Denote by $\hat{\est}$ the last estimate computed by $\cB_l$ prior to returning, setting it to $0$ if $\cB_l$ outputs $(0,\bot)$. It is ensured by the definition of $\cB_l$ that
    \begin{equation*}
        \abs{\est_o - \hat{\est}} < \tau.
    \end{equation*}
    If $\pi_l \neq \bot$, since the estimation $\hat{\est}$ is computed by averaging $q$ samples of $\cB_{l+1, 1}$, then by a Hoeffding bound
    \begin{equation*}
        \Prob{}{\abs{\hat{\est} - \Expec{\substack{ r_{l+1} \gets \cU_{m_{l+1}}}}{\cB_{l+1,1}(x; r_1, \pi_1, \dots, r_l, \pi_l, r_{l+1})}} > \tau}  < 2e^{-2\tau^2 q}. 
    \end{equation*}
    Further, the above inequality remains true in the case $\pi_l = \bot$, as in this case both quantities $\hat{\est}$ and the expectation of $\cB_{l+1,1}$ above are $0$.
    
    Thus, the output of $\cB_l$ is guaranteed to satisfy
    \begin{align*}
        &\Prob{(\est_o, \pi_l)\gets \cB_{l}(x;r_1,\pi_1, \dots, r_l)}{\abs{\est_o - \Expec{\substack{ r_{l+1} \gets \cU_{m_{l+1}}}}{\cB_{l+1,1}(x; r_1, \dots, \pi_l, r_{l+1})}} > 2\tau}\\ &\qquad  < 2e^{-2\tau^2 q}.
    \end{align*}
    As the outputs of $\cB_{l,1}$ and $\cB_{l+1,1}$ always lie in the range $[0,1]$, the above implies 
    \begin{align}
    \label{eq:crpc_bi}
        &\abs{\Expec{}{\cB_{l,1}(x; r_1, \pi_1, \dots, r_l)} - \Expec{\substack{ \pi_l \gets \cB_{l,2}(x; r_1, \pi_1, \dots, r_l)\\  r_{l+1} \gets \cU_{m_{l+1}}}}{\cB_{l+1,1}(x; r_1, \dots, \pi_l, r_{l+1})}} \notag \\&\quad < 2(\tau +  e^{-2t^2q}).
    \end{align}
    Therefore, we derive that 
    \begin{align*}
        &\abs{\cT_{l}(x; r_1, \pi_1, \dots, r_l)  - \Expec{}{\cB_{l,1}(x; r_1, \pi_1, \dots, r_l)}}\\ 
        &< \abs{ \Expec{\substack{\pi_l \gets \cB_{l,2}(x; r_1,\pi_1, \dots, r_l)\\ r_{l+1} \gets \cU_{m_{l+1}}}}{\cT_{l+1}(x; r_1, \dots, \pi_l, r_{l+1}) - \Expec{}{\cB_{l+1,1}(x; r_1, \dots, \pi_l, r_{l+1})}} } \\&\quad+ 2(\tau+ e^{-2t^2q})\\
        &\le \Expec{\substack{\pi_l \gets \cB_{l,2}(x; r_1,\pi_1, \dots, r_l)\\ r_{l+1} \gets \cU_{m_{l+1}}}}{\abs{\cT_{l+1}(x; r_1, \dots, \pi_l, r_{l+1}) - \Expec{}{\cB_{l+1,1}(x; r_1, \dots, \pi_l, r_{l+1})}}}\\&\quad+ 2(\tau+ e^{-2t^2q})\\
        &\le 2(k-l)(\tau + e^{-2t^2q}). 
    \end{align*}
    The first inequality follows from the triangle inequality and (\ref{eq:cr_true_rec}) and (\ref{eq:crpc_bi}). The last inequality holds since we assume that (\ref{eq:crpc_est_prot_i}) is true for $i=l+1$. Therefore, we find that (\ref{eq:crpc_est_prot_i}) holds for $i=l$, which completes the induction and further concludes that (\ref{eq:crpc_est_prot_i}) is satisfied for any $i\in [k]$.
    
    When $q =n \cdot p(n)^2$ and $ \tau = 1/p(n)$, by (\ref{eq:crpc_est_prot_i}), we obtain that, for any $x$ and $r_1$
    \begin{equation*}
        \abs{{\cT_{1}(x; r_1)} -\Expec{}{\cB_{1,1}(x; r_1)}} \le 2(k-1) (\tau + e^{-2t^2q}) \le 2(k-1) \left(\frac{1}{p(n)} + e^{-2n}\right). 
    \end{equation*}
    By combing the above with (\ref{eq:crpc_est_acc}) and (\ref{eq:crpc_real_acc}), we conclude
    \begin{align*}
        \abs{\Prob{}{\prot{\tsfP}{\sfV}(x) = 1} - \cB({x}) } &= \abs{\Expec{r_1 \gets \cU_{m_1}}{\cT_{1}(x;r_1)} -\Expec{r_1 \gets \cU_{m_1}}{\cB_{1,1}(x; r_1)}} \\&\le \frac{2(k-1)}{p(n)} + \negl(n).
    \end{align*}
\end{proofof}


\begin{proofof}{Claim \ref{claim:crpc_est_bd}}
    Let $w$ be a valid $\NP$ witness for $x$ satisfying $(x,w)\in\cR_\cL$. Consider the following distributions. 
    \begin{enumerate}
        \item $D_S$: sample $\rho \gets \cU_{\ell}$, $(r_1,\pi_1, \dots, r_k, \pi_k) \gets \Sim(x; \rho)$, $a \gets \sfV(x; r_1, \dots, \pi_k)$, output $(r_1, \pi_1, \dots, r_k, a)$
        \item $D_I$: sample $(r_1, \pi_1, \dots, r_k) \gets \tran{2k-1}(x,w)$, output $(r_1, \pi_1, \dots, r_k, 1)$
    \end{enumerate}
    For simplicity, denote the transcript $(r_1,\pi_1, \dots)$ sampled by $\tran{i}(x,w)$, which only contains the first $i$ messages in the protocol $\prot{\sfP_w}{\sfV}$ on $x$. 
    Analogous to \cref{eq:nizk_ecz}, by completeness, zero-knowledge, and data processing inequality, 
    \begin{equation*}
        \Delta_c(D_S; D_I) \le \ecn + \ezn .
    \end{equation*}
        Recall the algorithm $\cB_{k}$. 
    \begin{framed}
        \noindent \underline{$\cB_{k}(x; r_1,\pi_1, \dots, r_k)$}:
        \begin{enumerate}
            \item $\rho \gets \cA_{k}(r_1, \pi_1, \dots, r_k, 1)$
            \item If $f_{k,x}(\rho) = (r_1, \pi_1, \dots, r_k, 1)$
            \begin{itemize}
                \item $\pi_k \gets \Sim_{2k}(x; \rho)$
                \item Output $(1, \pi_k)$
            \end{itemize}
            \item Else output $(0, \bot)$
        \end{enumerate}
    \end{framed}
    Since the inverter $\cA_k$ satisfies (\ref{eq:crpc_inv_k}), and the transcript on which $\cB_k$ queries $\cA_k$ is the same as $D_I$ when $(r_1,\pi_1, \dots, r_k)$ is generated by the protocol, by the data processing inequality \pnote{Could add a few more steps below}  
    \begin{align}
    \label{eq:crpc_bd_1}
        \Expec{(r_1,\pi_1, \dots, r_k) \gets \tran{2k-1}(x)}{\cB_{k,1}(x; r_1, \pi_1, \dots, r_k)} &\ge 1 - \frac{1}{p(n)} - \Delta_c(D_S; D_I)\notag \\&\ge 1 -\frac{1}{p(n)} - \ecn - \ezn .
    \end{align}
    For each $i\in [k]$, we consider the following values.
    \begin{equation*}
        \Expec{(r_1,\pi_1, \dots, r_i) \gets \tran{2i-1}(x)}{\cB_{i,1}(x; r_1, \pi_1, \dots, r_i)}  
    \end{equation*}
    With the observation from (\ref{eq:crpc_est_prot_i}), the above value is somewhat related to the acceptance probability when a prover runs the ZK protocol with the honest prover's strategy $\sfP$ for the first $(i-1)$ rounds and our inverter-based prover $\tsfP$ for the rest of the rounds. 
    
    Next, we show the connection between cases $i$ and $i+1$. In particular, we provide a lower bound $\mathbb{E}[{\cB_{i,1}}]$ using $\mathbb{E}[{\cB_{i+1,1}}]$. Let $\mathsf{Est}_i$ denote the estimation procedure of the expectation of $\cB_{i+1,1}$. More specifically, 
    \begin{framed}
        \noindent \underline{$\Est_{i}(x; r_1, \pi_1, \dots, r_i, \pi_i; \sigma_1, \dots, \sigma_{q})$}:
        \begin{enumerate}
            \item $\est \gets \frac{1}{q} \sum_{i \in [q]} \cB_{i+1, 1}(x; r_1, \pi_1, \dots, r_i, \pi_i, \sigma_j)$
            \item Output $\est$
        \end{enumerate}
    \end{framed}
    The input $(\sigma_1, \dots, \sigma_q)$ can be viewed as randomness. 
    For simplicity, we sometimes omit $(\sigma_1, \dots, \sigma_{q})$ in the input to describe a randomized algorithm $\Est_{i}(x; r_1, \pi_1, \dots, r_i, \pi_i)$ where $\sigma_1, \dots, \sigma_{q}$ are drawn uniformly randomly. 
    Recall that $\cB_i$ works as follows. 
    \begin{framed}
        \noindent \underline{$\cB_{i}(x; r_1,\pi_1, \dots, r_i)$}:
        \begin{enumerate}
            \item For $a = q^{k-i}, q^{k-i}-1, \dots, 2, 1$ in decreasing order
            \begin{itemize}
                \item $(\rho, \sigma_{1}, \dots, \sigma_{q}, \rand_1, \dots, \rand_q) \gets \cA_{i}(r_1, \pi_1, \dots, r_i, a/q^{k-i})$
                \item If $f_{i,x}(\rho, \sigma_{1}, \dots, \sigma_{q}, \rand_1, \dots, \rand_q) = (r_1, \pi_1, \dots, r_i, a/q^{k-i})$
                \begin{itemize}
                    \item $\hat{\pi}_i \gets \Sim_{2i}(x; \rho)$
                    \item $\hat{\est} \gets \Est_i(x; r_1,\dots, r_1,\pi_1, \dots, r_i, \hat{\pi}_i)$
                    \item If $\abs{\hat{\est} - a/q^{k-i}}< \tau$
                    \begin{itemize}
                        \item Output $(a/q^{k-i}, \hat{\pi}_i)$
                    \end{itemize}
                \end{itemize}
            \end{itemize}
            \item Output $(0,\bot)$
        \end{enumerate}
    \end{framed}
    Additionally, define an algorithm $\hat{\cB}_i$ as follows that almost simulates the procedure of $\cB_i$ in a single iteration except for the final check. 
    \begin{framed}
        \noindent \underline{$\hat{\cB}_{i}(x; r_1, \pi_1, \dots, r_i, \est)$}:
        \begin{enumerate}
            \item $(\rho, \sigma_{1}, \dots, \sigma_{q}, \rand_1, \dots, \rand_q) \gets \cA_{i}(r_1, \pi_1, \dots, r_i, \est)$
            \item If $f_{i,x}(\rho, \sigma_{1}, \dots, \sigma_{q}, \rand_1, \dots, \rand_q) = (r_1, \pi_1, \dots, r_i, \est)$
            \begin{itemize}
                \item $\hat{\pi}_i \gets \Sim_{2i}(x; \rho)$
                \item $\hat{\est} \gets \Est_{i}(x; r_1, \pi_1, \dots, r_i, \hat{\pi}_i)$
                \item Output $\hat{\est}$
            \end{itemize}
            \item Output $\bot$
        \end{enumerate}
    \end{framed}
    Denote by $E(\est, \hat{\est})$ the event that the distance between two inputs are close. Let $E(\est, \hat{\est}) = 0$ if the input contains $\bot$. Otherwise,  
    \begin{align*}
        E(\est, \hat{\est}) &= \mathbf{1}\left[\abs{\est - \hat{\est}}< \tau\right]. 
    \end{align*}
    $\cB_{i,1}(x; r_1, \pi_1, \dots, r_i)$ basically iterates $\hat{\cB}_i(x; r_1, \pi_1, \dots, r_i, \est)$ to obtain $\hat{\est}$, from $\est = 1, 1 - 1/q^{k-i}, \dots$ to $1/q^{k-i}$ in decreasing order, and outputs the largest $\est$ when it finds that $E(\est, \hat{\est}) = 1$. 
    Then, for any $(x; r_1, \pi_1, \dots, r_i)$ and any $\est$, the following is true since, with probability at least $\Prob{\hat{\est}}{|\est - \hat{\est}|< \tau}$, $\cB_{i,1}$ outputs a value at least $\est$.
    \begin{equation}
    \label{eq:crpc_bd_i}
        \Expec{}{\cB_{i, 1}(x; r_1, \pi_1, \dots, r_i)} \ge \Expec{ \hat{\est} \gets \hat{\cB}_{i}(x; r_1, \pi_1, \dots, r_i, \est)}{\est \cdot E(\est, \hat{\est})}.
    \end{equation}
    \ynote{Haven't got a better idea to explain why this is true...} 
    In the following, we show a lower bound for the right-hand side considering $(r_1,\pi_1, \dots, r_i)$ is generated by the protocol. 

    The probability that $\est'$ fails the event $E(\est, \est')$ is only influenced by two parts: the performance of the inverter $\cA_i$ and the estimation of $\cB_{i+1}$. Next, we lower-bound the probability that $E(\est, \est')=1$ on some input distribution that is of interest to us.
    Consider the following procedure $\cP$: 
    \begin{framed}
        \noindent \underline{$\cP_{i,x}((\rho, \sigma_{1}, \dots, \sigma_{q}, \rand_1, \dots, \rand_q), (r_1, \pi_1, \dots, r_i, \est))$}:
        \begin{enumerate}
            \item If $f_{i,x}(\rho, \sigma_{1}, \dots, \sigma_{q}, \rand_1, \dots, \rand_q) = (r_1, \pi_1, \dots, r_i, \est)$
            \begin{itemize}
                \item $\hat{\pi}_i \gets \Sim_{2i}(x; \rho)$
                \item $\hat{\est} \gets \Est_i(x; r_1,\pi_1, \dots, r_i, \hat{\pi}_i)$
                \item Output $(\est, \hat{\est})$
            \end{itemize}
            \item Else output $(\est, \bot)$
        \end{enumerate}
    \end{framed}
    Let $\est$ and $\hat{\est}$ be sampled from $\Est_i$ independently, by Chernoff bounds, both estimation values are concentrated around the expectation, then with only small probability, these two values differ significantly 
    \begin{align}
    \label{eq:crpc_est_tail}
        \Prob{\substack{(r_1, \pi_1, \dots, r_i, \pi_i) \gets \Sim_{1\dots 2i}(x) \\ \est,\hat{\est} \gets \Est_{i}(x; r_1,\pi_1, \dots, r_i,\pi_i)}}{\abs{\est - \hat{\est}}> \tau} < 4 e^{-\tau^2q/2}, 
    \end{align}
    where the distribution $(\est, \hat{\est})$ sampled from is the same as $\mathcal{P}_{i,x}(\cU, f_{i,x}(\cU))$. When $q = n \cdot p(n)^2$ and $\tau = 1/ p(n)$, the above value is at most $4e^{-n/2}$, negligible in $n$. \ynote{Change to $\mathcal{P}$}
    
    Notice that the following two distributions are equivalent.
    \begin{align*}
        \cP_{i,x}(\cA_{i}(f_{i,x}(\cU)), f_{i,x}(\cU))
        =
        \left\{
        \begin{array}{c}
         (r_1, \pi_1, \dots, r_i, \pi_i) \gets \Sim_{1\dots 2i}(x)\\
         \est  \gets \Est_i(x; r_i, \pi_1, \dots, r_i, \pi_i)  \\
         \hat{\est} \gets \hat{\cB}_{i}(x; r_1, \pi_1, \dots, r_i, \est)\\
        \text{output }(\est, \hat{\est})
        \end{array}
        \right\}
    \end{align*}
    By our assumption (\ref{eq:crpc_inv_k}), the performance of $\cA_{i}$ ensures that
    \begin{equation*}
        \Delta_s(\cU, f_{i,x}(\cU); \cA_{i}(f_{i,x}(\cU)), f_{i,x}(\cU)) < \frac{1}{p(n)},
    \end{equation*}
    which implies that 
    \begin{align*}
        &\Expec{\substack{(r_1, \dots, \pi_i) \gets \Sim_{1\dots 2i}(x)\\
         \est  \gets \Est_{i}(x; r_1, \pi_1, \dots, r_i, \pi_i)  \\
         \hat{\est} \gets \hat{\cB}_{i}(x; r_1, \pi_1, \dots, r_i, \est)}}{E(\est, \hat{\est})}\\ &> \Expec{\substack{(r_1, \pi_1, \dots, r_i, \pi_i) \gets \Sim_{1\dots 2i}(x) \\ \est,\hat{\est} \gets \Est_{i}(x; r_1,\pi_1, \dots, r_i,\pi_i)}}{E(\est, \hat{\est})} - \frac{1}{p(n)} \\ &> 1 - \frac{1}{p(n)} - \negl(n),
    \end{align*}
    where the first inequality is obtained by the data processing inequality and (\ref{eq:crpc_inv_k}) and the second follows (\ref{eq:crpc_est_tail}). 
    Since $\Delta_c(\Sim(x); \tr(x)) \le \ezn$, when $(r_1, \dots, \pi_i)$ is sampled by the protocol 
    \begin{equation*}
        \Expec{\substack{(r_1, \pi_1, \dots, r_i, \pi_i) \gets \tran{2i}(x)\\
         \est  \gets \Est_{i}(x; r_1, \pi_1, \dots, r_i, \pi_i)  \\
         \hat{\est} \gets \hat{\cB}_{i}(x; r_1, \pi_1, \dots, r_i, \est)}}{E(\est, \hat{\est})} > 1- \ezn - \frac{1}{p(n)} - \negl(n). 
    \end{equation*}
    Then, by (\ref{eq:crpc_bd_i}), we obtain
    \begin{align}
    \label{eq:crpc_bd_mid}
         &\Expec{(r_1, \pi_1, \dots, r_i) \gets \tran{2i-1}(x) }{\cB_{i, 1}(x; r_1, \pi_1, \dots, r_{i})}\\&\qquad \ge \Expec{\substack{(r_1, \pi_1, \dots, r_i, \pi_i) \gets \tran{2i}(x) \\ 
         \est  \gets \Est_{i}(x; r_1, \pi_1, \dots, r_i, \pi_i) \notag \\
         \hat{\est} \gets \hat{\cB}_{i}(x; r_1, \pi_1, \dots, r_i, \est)}}{\est \cdot E(\est, \hat{\est})}\\ 
        &\qquad \ge \Expec{\substack{(r_1, \dots, r_i, \pi_i) \gets \tran{2i}(x) \\ 
         \est  \gets \Est_{i}(x; r_i, \dots, r_i, \pi_i)}}{\est} - \ezn - \frac{1}{p(n)}- \negl(n)\notag,
    \end{align}    
    where the last inequality is obtained because $\est$ is valued in $[0,1]$, $\mathbb
    {E}[\est \cdot E(\est, \hat{\est})] \ge \mathbb{E}[\est] - \Pr[E(\est, \hat{\est}) = 0]$.
    By the definition of $\Est_i$, for any input $(x; r_1, \pi_1, \dots, r_i, \pi_i)$,
    \begin{equation*}
         \Expec{\substack{
         \est  \gets \Est_{i}(x; r_1, \pi_1, \dots, r_i, \pi_i)  }}{\est}= \Expec{\substack{ r_{i+1} \gets \cU_{m_{i+1}} }}{\cB_{i+1, 1}(x; r_1, \pi_1, \dots, r_i, \pi_i, r_{i+1})}.
    \end{equation*}
    Therefore, we relate the expected values of $\cB_{i}$ and $\cB_{i+1}$ by the following
    \begin{align*}
        (\ref{eq:crpc_bd_mid})\quad 
        &\ge \Expec{(r_1, \dots,\pi_i, r_{i+1}) \gets \tran{2i+1}(x)}{\cB_{i+1,1}(x; r_1, \dots, \pi_i, r_{i+1})} - \ezn - \frac{1}{p(n)} - \negl(n).
    \end{align*}
   Combing (\ref{eq:crpc_bd_1}) and above, we have 
    \begin{align*}
        &\Expec{(r_1, \dots, r_{i}) \gets \tran{2i+1}(x)}{\cB_{i,1}(x; r_1, \dots, r_{i})} \\& \qquad\ge 1- \ecn - (k-i+1) \left(\ezn + \frac{1}{p(n)}\right) - \negl(n).
    \end{align*}
    % ...
    Therefore, let $i=1$
    \begin{equation*}
        \cB(x) = \Expec{r_1 \gets \cU_{m_1}}{\cB_{1,1}(x; r_1)} \ge 1 - \ecn - k \cdot \ezn - \frac{k}{p(n)} - \negl(n), 
    \end{equation*}
    which concludes the proof. 
\end{proofof}

%% file: owfbig.tex
\section{One-Way Functions}\label{sec:owf}

We complete our transformation by extending the result of~\cite{DBLP:conf/stoc/HiraharaN24} that shows that assuming there are zero-knowledge arguments for NP, auxiliary-input one-way functions imply standard one-way functions. We follow the approach taken by~\cite{DBLP:conf/crypto/ChakrabortyHK25} in the context of weak zero-knowledge arguments, showing that the auxiliary-input OWFs obtained in the previous sections also imply standard one-way functions. The resulting theorems are as follows.

\begin{theorem}
\label{the:nizk_owf}
    If $\NP \not \subseteq \mathsf{ioP/poly}$ and, for some $\ec, \es, \ez: \Nat \rightarrow [0,1]$ such that $\ec+\es+\ez <_n 1$, every language in $\NP$ has an $(\ec, \es, \ez)$-NIZK proof, then one-way functions exist. The same holds for NIZK arguments.
\end{theorem}

\begin{theorem}
\label{the:pczk_owf}
    If $\NP \not \subseteq \mathsf{ioP/poly}$ and, for some $\ec, \es, \ez: \Nat \rightarrow [0,1]$ and $t:\Nat\ra\Nat$ such that $\ec+\es+(t-1)\cdot \ez <_n 1$, every language in $\NP$ has a $t$-message public-coin $(\ec, \es, \ez)$-ZK proof, then one-way functions exist. The same holds for ZK arguments.
\end{theorem}

\begin{theorem}
\label{the:cr_owf}
    If $\NP \not \subseteq \mathsf{P/poly}$ and, for some $\ec, \es, \ez: \Nat \rightarrow [0,1]$ and constant $k\in\Nat$ such that $\ec+\es+k\cdot \ez <_n 1$, every language in $\NP$ has a $k$-round public-coin $(\ec, \es, \ez)$-ZK proof, then infinitely often one-way functions exist. The same holds for ZK arguments.
\end{theorem}

The proofs of these theorems follow from combining the respective lemmas in \cref{sec:aiowf} with lemmas stated later in this section. \cref{the:nizk_owf,the:pczk_owf} are proven in \cref{sec:nizk-pc-proof}, and \cref{the:cr_owf} in \cref{sec:const-proof}.

As observed in \cite[Section 5]{DBLP:conf/crypto/ChakrabortyHK25} (which in turn draws on the approach in \cite{DBLP:journals/iacr/LiuMP24}), the overall implication from weak ZK to OWFs can be broken up into three parts. The first part consists of showing that having such a weak ZK protocol for worst-case hard language implies an auxiliary-input OWF, and this we have shown in \cref{sec:aiowf} (for various flavors of ZK protocols). The two remaining steps are the following: 
\begin{enumerate}
    \item Show that auxiliary-input OWFs imply that there exist what are called in~\cite{DBLP:conf/crypto/ChakrabortyHK25} as \emph{one-sided average-case hard} languages. 
    \item Show that (the corresponding flavor of) weak ZK protocols for such one-sided average-case hard languages imply standard OWFs. 
\end{enumerate}

We outline these transformations below for each flavor of weak ZK that we have considered so far. Some of the proofs below are along the lines of those of similar lemmas from prior work, especially from \cite{DBLP:conf/crypto/ChakrabortyHK25}.

\subsection{One-Sided Average-Case Hardness from ai-OWF}

This lemma appears in~\cite[Lemma 4]{DBLP:conf/crypto/ChakrabortyHK25}. This result is somewhat independent, and does not have anything to do with weak zero-knowledge. We can thus use it as is.  

\begin{definition}[One-Sided Average-Case BPP]
    The class $\oabpp$ consists of average-case problems $(\cL, \cD)$ for which there is a non-uniform probabilistic polynomial-time algorithm $\cA$ and functions $a,b: \Nat \rightarrow [0,1]$ with $a >_n b$, such that for all sufficiently large $n$
    \begin{enumerate}
        \item For every $x\in \cL_n$, $\Prob{}{\cA(x) = 1} \ge a(n)$. 
        \item $\Prob{x\gets \cD_n}{\cA(x) = 1| x\notin \cL_n} \le b(n)$. 
    \end{enumerate}
\end{definition}

    

\begin{lemma}[{\cite[Lemma 4]{DBLP:conf/crypto/ChakrabortyHK25}}]
\label{lem:aiowf_ioosac}
    Suppose that there exists an auxiliary-input one-way function. Then there exists an $\mathcal{L} \in \mathsf{NP}$ such that $(\mathcal{L}, \cU) \notin \mathsf{io1AvgBPP}/\mathsf{poly}$ and $\Prob{x\gets \cU_n}{x\notin \cL_n} \ge 1/2$ for all $n\in \Nat$.
\end{lemma}

The following variant also follows from the proof of the above lemma.\pnote{To be verified}

\begin{lemma}
\label{lem:aiowf_osac}
    Suppose that there exists an infinitely-often ai-OWF, then there exists a language $\cL\in \NP$ such that $(\cL, \cU) \notin \oabpp$ and $\Prob{x\gets \cU_n}{x\notin \cL_n} \ge 1/2$ for all $n\in \Nat$. 
\end{lemma}






\subsection{OWF from One-Sided Average-Case Hardness}

\subsubsection{NIZK and Public-Coin ZK}
\label{sec:nizk-pc-proof}


Combining Lemma \ref{lem:nizk_red} and Lemma \ref{lem:pczk_red} with the amplification for converting weak and distributional one-way functions into standard form~\cite{DBLP:conf/focs/Yao82a,DBLP:conf/focs/ImpagliazzoL89}, we obtain the following lemma.   
\begin{lemma}
\label{lem:zk_red}
    For some $\ec, \es, \ez: \Nat \rightarrow [0,1]$ and $t:\Nat \rightarrow \Nat$, consider a language $\cL\in\NP$ with an $(\ec,\es, \ez)$-NIZK protocol (in which case $t=2$) or a $t$-message $(\ec,\es, \ez)$-public-coin ZK protocol. 
    For any polynomials $p_1, p_2$, there is a reduction $R$, which is a polynomial-time oracle-aided algorithm, and a polynomial-time computable function family $\cF = \set{f_x}_{x\in \zo^*}$ such that for any probabilistic polynomial-time algorithm $\cA$ and all large enough $n$,
    \begin{enumerate}
        \item When $x\in \cL_n$, if $\cA$ inverts $f_x$ with probability at least $1/p_1(n)$, then 
        \begin{equation*}
            \Prob{}{R^{\cA}(x) = 1} \ge 1 - \ecn - (t(n)-1)\cdot \ezn - 1/p_2(n). 
        \end{equation*}
        \item When $x \in \zo^n\setminus\cL_n$, 
        \begin{equation*}
            \Prob{}{R^{\cA}(x) = 1} \le \esn.
        \end{equation*}
    \end{enumerate}
\end{lemma}
Note that our constructions of $R$ and $\cF$ in the proofs of Lemma \ref{lem:nizk_red} and Lemma \ref{lem:pczk_red} do not satisfy the above conditions. However, reductions and functions satisfying the above lemma can be obtained from the previous construction by applying the techniques in~\cite{DBLP:conf/focs/Yao82a,DBLP:conf/focs/ImpagliazzoL89}. 
Next, we use the following analogue of~\cite[Lemma 5]{DBLP:conf/crypto/ChakrabortyHK25}. 

\begin{lemma}
\label{lem:zk_owf}
    For $\ec, \es, \ez: \Nat \rightarrow [0,1]$ and $t:\Nat \rightarrow \Nat$, consider any language $\cL$ such that $(\cL,\cU) \notin \iooabpp$, and $\Prob{x\gets \cU_n}{x\notin 
    \cL_n} \ge 1/2$ for all $n\in \Nat$. If there is an $(\ec, \es, \ez)$-NIZK protocol (in which case $t=2$) or a $t$-message $(\ec, \es,\ez)$-public-coin ZK protocol for $\cL$ with $\ec+\es+ (t-1)\cdot \ez <_n 1$, then one-way functions exist.        
\end{lemma}



\begin{proofof}{Lemma \ref{lem:zk_owf}}
    Since $\ec+\es+(t-1) \cdot \ez <_n 1$, there exists a polynomial $q$ such that
    \begin{equation}
    \label{eq:zk_owf_ineq}
        \ecn + \esn + (t(n)-1)\cdot \ezn + \frac{1}{q(n)} < 1 
    \end{equation}
    holds for all sufficiently large $n \in \Nat$. Then, let $p = 4q$.
    Let $\cF = \set{f_x}_{x\in \zo^*}$ be the construction in Lemma \ref{lem:zk_red}, and let $f(x,r) = (x,f_x(r))$. 

    Suppose that there is a (non-uniform) PPT algorithm $\cA$ that inverts $f$ for infinitely many $n\in \Nat$ with probability $(1- 1/p(n))$. 
    In particular, we assume that 
    \begin{equation}
    \label{eq:break_owf}
        \Prob{\substack{ (x,y) \gets f(\cU)}}{f(\cA(x, y)) = (x,y)} \ge 1 - \frac{1}{p(n)},
    \end{equation}
    where we denote by $f(\cU)$ the distribution of the output of $f$ when the input is sampled uniformly. Denote by $\cA_x(y)$ the algorithm: on input $y$, it runs $(\hat{x}, \hat{r})\gets \cA(x, y)$; if $\hat{x}$ matches $x$, it outputs $\hat{r}$; else outputs $\bot$. 
    Then, we construct an algorithm $\mathcal{C}$ as follows. 
    \begin{framed}
        \noindent \underline{Algorithm $\cC(x)$}:
        \begin{enumerate}
            \item Sample $y \gets f_x(\cU)$
            \item Let $\hat{r} \gets \cA_x(y)$
            \item If $f_x(\hat{r}) = y$, output $R^{\cA}(x)$
            \item Else output $1$
        \end{enumerate}
    \end{framed}
    Consider $x\in \cL_n$, if $\cA_x$ correctly inverts $f_x$ on random input $r$ with probability at least $1/p(n)$, then by Lemma \ref{lem:zk_red}, we have that 
    \begin{equation*}
        \Prob{}{\cC(x) = 1} \ge 1 - \Prob{}{R^{\cA}(x) \neq 1} \ge 1 - \ecn - (t(n)-1)\cdot \ezn - \frac{1}{p(n)}. 
    \end{equation*}
    If the probability that $\cA_x$ finds the correct pre-image of $f_x(r)$ is upper-bounded by $1/p(n)$, 
    \begin{equation*}
        \Prob{}{\cC(x) = 1} \ge 1 - \Prob{\substack{ y \gets f_x(\cU)}}{f_x(\cA_x(y)) = y} \ge 1 - \frac{1}{p(n)}. 
    \end{equation*}
    Then, the probability that $\cC(x)$ outputs $1$ is at least 
    \begin{equation*}
        \Prob{}{\cC(x) = 1} \ge 1 - \ecn - (t(n)-1)\cdot \ezn - \frac{1}{p(n)}.
    \end{equation*}
    Consider the case $x\gets \zo^n\setminus \cL_n$,  
    \begin{align*}
        &\Prob{\substack{x\gets \cU_n \\ y \gets f_x(\cU)}}{f(\cA(x, y)) \neq (x, y)| x\notin \cL_n}\cdot \Prob{x \gets \cU_n}{x\notin \cL_n} \\&\le \Prob{\substack{x\gets \cU_n \\ y \gets f_x(\cU)}}{f(\cA(x,y)) \neq (x,y)}\\ &\le \frac{1}{p(n)},
    \end{align*}
    where the last inequality holds because (\ref{eq:break_owf}), 
    and we make an assumption on the language 
    \begin{equation*}
        \Prob{x\gets \cU_n}{x\notin \cL_n} \ge \frac{1}{2}
    \end{equation*}
    then we obtain that
    \begin{equation*}
        \Prob{\substack{x\gets \cU_n\\ y \gets f_x(\cU)}}{f_x(\cA_x(y)) \neq y| x\notin \cL_n} = \Prob{\substack{x\gets \cU_n \\ y \gets f_x(\cU)}}{f(\cA(x, y)) \neq (x, y)| x\notin \cL_n} \le \frac{2}{p(n)}. 
    \end{equation*}
    Therefore, by a union bound
    \begin{align*}
        &\Prob{x\gets \cU_n}{\cC(x) = 1 | x\notin \cL_n} \\&\le \Prob{\substack{x\gets \cU_n\\ y \gets f_x(\cU)}}{f_x(\cA_x(y)) \neq y| x\notin \cL_n} + \Prob{x\gets \cU_n}{R^{\cA}(x) = 1 | x\notin \cL_n} \\
        &\le \frac{2}{p(n)} + \esn. 
    \end{align*}

    $\cC$ is an algorithm that decides the language $\cL$ in the one-sided average case, since by (\ref{eq:zk_owf_ineq})
    \begin{equation*}
        \frac{2}{p(n)} + \esn <_n 1 - \ecn - (t(n)-1)\cdot \ezn - \frac{1}{p(n)}
    \end{equation*}
    which contradicts $\cL \notin \iooabpp$. Therefore, $f$ is a weak OWF and from \cite{DBLP:conf/focs/Yao82a}, it implies the existence of one-way functions. 
\end{proofof}

\medskip
\paragraph{Proving \cref{the:nizk_owf,the:pczk_owf}}
Combining \cref{lem:nizk_aiowf}, \cref{lem:aiowf_ioosac} and \cref{lem:zk_owf} yields \cref{the:nizk_owf}.
Analogously, for any round public-coin ZK protocol, by putting \cref{lem:pczk_aiowf}, \cref{lem:aiowf_ioosac} and \cref{lem:zk_owf} together, we obtain \cref{the:pczk_owf}. 


\subsubsection{Constant-Round Public-Coin ZK}
\label{sec:const-proof}

Lemma \ref{lem:cr_red} combined with~\cite{DBLP:conf/focs/ImpagliazzoL89} yields the following lemma. 
\begin{lemma}
\label{lem:ioowf_red}
    For some $\ec, \es, \ez: \Nat \rightarrow [0,1]$ and any constant $k\in \Nat$, suppose a language $\cL\in \NP$ has a $k$-round public-coin $(\ec, \es, \ez)$-ZK proof or argument. For any polynomials $p_1, p_2$, there exists a polynomial-time oracle-aided algorithm $R$ and families of oracle-aided functions $\cF_1, \dots, \cF_k$ satisfying the following. 
    
    For $i\in[k]$, the family $\cF_{i} = \set{f_{i,x}}_{x\in\zo^*}$ consists of functions that require access to $(k-i)$ oracles, and are polynomial-time computable given these oracles. For any sequence of probabilistic polynomial-time algorithms $\cA_1,\dots, \cA_k$ and all large enough $n$, we have the following.
    
    \begin{enumerate}
        \item For every $x\in \cL_n$, if for all $i\in [k]$, $\cA_i$ inverts $f_{i,x}$ with probability at least $1/p_1(n)$, then 
        \begin{equation*}
            \Prob{}{R^{\cA_1\dots \cA_k}(x) = 1} \ge 1 - \ecn - k\cdot \ezn - 1/p_2(n)
        \end{equation*}
        \item For every $x\in \zo^n \setminus \cL_n$, 
        \begin{equation*}
            \Prob{}{R^{\cA_1\dots \cA_k}(x) = 1} \le \esn. 
        \end{equation*}
    \end{enumerate}
\end{lemma}

\begin{lemma}
\label{lem:zk_ioowf}
    For $\ec, \es, \ez: \Nat \rightarrow [0,1]$ and a constant $k\in \Nat$, consider any language $\cL$ such that $(\cL,\cU) \notin \oabpp$, and $\Prob{x\gets \cU_n}{x\notin 
    \cL_n} \ge 1/2$ for all $n\in \Nat$. If there is a $k$-round $(\ec, \es,\ez)$-public-coin ZK protocol for $\cL$ with $\ec+\es+ k\cdot \ez <_n 1$, then infinitely-often one-way functions exist.
\end{lemma}


\begin{proofof}{Lemma \ref{lem:zk_ioowf}}
    Since $\ec+\es+k \cdot \ez <_n 1$, there exists a polynomial $q$ such that
    \begin{equation}
    \label{eq:zk_ioowf_ineq}
        \ecn + \esn + k \cdot \ezn + \frac{1}{q(n)} < 1 
    \end{equation}
    holds for all sufficiently large $n \in \Nat$. Then, let $p = 2(k+1)q$.
    Let $\cF_1, \dots, \cF_k$ be the constructions from Lemma \ref{lem:zk_red}. 
    Suppose that there are (non-uniform) PPT algorithms $\cA_1, \dots, \cA_k$ that invert the corresponding function for all sufficiently large $n\in \Nat$. 
    For simplicity, we omit the superscript that describes the oracle access in the function constructions, denote $\cF_i = \set{f_{i,x}}$ for $i\in [k]$ and let $f_i(x,r) = (x,f_{i,x}(r))$. The efficiency conditions for the algorithms $\cA_i$'s implies that all functions $f_i$'s are polynomial-time computable. 
    In particular, we assume that for each $i\in [k]$, 
    \begin{equation*}
        \Prob{\substack{ (x,y) \gets f_{i}(\cU)}}{f_{i}(\cA_i(x, y)) = (x, y)} \ge 1 - \frac{1}{p(n)}. 
    \end{equation*}
    Denote by $\cA_{i,x}(y)$ the algorithm: on input $y$, it runs $(\hat{x}, \hat{r})\gets \cA_i(x, y)$; if $\hat{x}$ matches $x$, it outputs $\hat{r}$; else outputs $\bot$. 
    In the following, we construct an algorithm $\mathcal{C}$ towards solving the language $\cL$ in the one-sided average case. 
    \begin{framed}
        \noindent \underline{Algorithm $\cC(x)$}:
        \begin{enumerate}
            \item For $i \in [k]$
            \begin{itemize}
                \item Sample $y \gets f_{i,x}(\cU)$
                \item Set $\hat{r}_i \gets \cA_{i,x}(y)$
                \item If $f_{i,x}(\hat{r}_i) \neq y$, output $1$
            \end{itemize}
            \item Output $R^{\cA_{1}\dots\cA_{k}}(x)$
        \end{enumerate}
    \end{framed}
    Consider $x\in \cL_n$, if for all $i\in [k]$, $\cA_{i,x}$ correctly inverts $f_{i,x}$ on random input $r$ with probability at least $1/p(n)$, then by Lemma \ref{lem:ioowf_red}, we have that 
    \begin{equation*}
        \Prob{}{\cC(x) = 1} \ge \Prob{}{R^{\cA_{1}\dots\cA_{k}}(x) = 1} \ge 1 - \ecn - k \cdot \ezn - \frac{1}{p(n)}. 
    \end{equation*}
    If there is an $i\in [k]$ such that the probability that $\cA_{i,x}$ finds the correct pre-image of $f_{i,x}(r)$ is upper-bounded by $1/p(n)$, 
    \begin{equation*}
        \Prob{}{\cC(x) = 1} \ge 1 - \Prob{\substack{ y \gets f_{i,x}(\cU)}}{f_{i,x}(\cA_{i,x}(y)) = y} \ge 1 - \frac{1}{p(n)}. 
    \end{equation*}
    Therefore, the probability that $\cC(x)$ outputs $1$ is at least 
    \begin{equation*}
        \Prob{}{\cC(x) = 1} \ge 1 - \ecn - k\cdot \ezn - \frac{1}{p(n)}.
    \end{equation*}
    When $x\in \zo^n\setminus \cL_n$, by our assumption on the algorithm, for each $i\in [k]$
    \begin{align*}
        &\Prob{\substack{(x, y) \gets f_i(\cU)}}{f_i(\cA(x, y)) \neq (x, y)| x\notin \cL_n}\cdot \Prob{x \gets \cU_n}{x\notin \cL_n}\\ &\quad \le \Prob{\substack{(x, y) \gets f_i(\cU)}}{f_i(\cA(x,y)) \neq (x,y)}\\ &\quad \le \frac{1}{p(n)}
    \end{align*}
    and the language satisfying 
    \begin{equation*}
        \Prob{x\gets \cU_n}{x\notin \cL_n} \ge \frac{1}{2}
    \end{equation*}
    we obtain
    \begin{equation*}
        \Prob{\substack{x\gets \cU_n\\ y \gets f_{i,x}(\cU)}}{f_{i,x}(\cA_{i,x}(y)) \neq y| x\notin \cL_n}
        \le \frac{2}{p(n)}. 
    \end{equation*}
    Therefore, by a union bound
    \begin{align*}
        &\Prob{x\gets \cU_n}{\cC(x) = 1 | x\notin \cL_n} \\&\le \sum_{i\in [k]}  \Prob{\substack{x\gets \cU_n\\ y \gets f_{i,x}(\cU)}}{f_{i,x}(\cA_{i,x}(y)) \neq y| x\notin \cL_n} + \Prob{x\gets \cU_n}{R^{\cA_{1}\dots\cA_{k}}(x) = 1 | x\notin \cL_n} \\
        &\le \frac{2k}{p(n)} + \esn. 
    \end{align*}
    $\cC$ is an algorithm that decides the language $\cL$ in the one-sided average case as by (\ref{eq:zk_ioowf_ineq}),
    \begin{equation*}
        \frac{2k}{p(n)} + \esn <_n 1 - \ecn - k\cdot \ezn - \frac{1}{p(n)}
    \end{equation*}
    which contradicts $\cL \notin \iooabpp$. Thus, at least one of the recursively defined $f_i$'s is a weak OWF and by \cite{DBLP:conf/focs/Yao82a}, one-way functions exist. 
\end{proofof}


\medskip
\paragraph{Proving \cref{the:cr_owf}}
\cref{the:cr_owf} follows directly from \cref{lem:cr_aiowf}, 
\cref{lem:aiowf_osac}, and \cref{lem:zk_ioowf}.






    


%% file: shortappendix.tex
\section{On Randomized Verification}\label{sec:randomizedverifier}

\rnote{We should mention in the lemma that we are considering probabilistic verifier decision. The main definition now considers only deterministic verifiers.}\ynote{Make sense}
\begin{lemma}
\label{lem:nizk_red_rd}
    For $\ec, \es, \ez: \Nat \rightarrow [0,1]$, suppose a language $\cL\in \NP$ has an $(\ec, \es, \ez)$-NIZK proof or argument (with randomized verification). Then there exists a reduction $R$, which is a polynomial-time oracle-aided algorithm, and a function family $\cF = \set{f_x}_{x\in \zo^*}$, such that for any probabilistic polynomial-time algorithm $\cA$, any polynomial $p$, and all large enough $n \in \Nat$
    \begin{enumerate}
        \item\label{it:nizk_yes_rd} For every $x\in \cL_n$, if $\cA$ distributionally inverts $f_x$ with deviation at most $1/{p(n)}$, then
        \begin{equation*}
            \Prob{}{R^{\cA}(x) = 1} \ge 1 - \ecn - \ezn - \frac{3}{p(n)} - \negl(n). 
        \end{equation*}
        \item\label{it:nizk_no_rd} For every $x\in \zo^n \setminus \cL_n$, 
        \begin{equation*}
            \Prob{}{R^{\cA}(x) = 1} \le \esn. 
        \end{equation*}
    \end{enumerate}
\end{lemma}

\begin{proof-sketch-of}{Lemma \ref{lem:nizk_red_rd}}
    Let $\Sim$ be the simulator that satisfies the zero-knowledge requirement. 
    Let $q = n \cdot p(n)^2$. 
    Let $\cF= \set{f_x}_{x \in \zo^*}$ be a function family where $f_x$ is defined as 
    \begin{framed}
        \noindent \underline{$f_x(\rho; \sigma_1,\dots, \sigma_q)$}:
        \begin{enumerate}
            \item $(r,\pi) \gets \Sim(x; \rho)$
            \item $\est \gets \frac{1}{q}\sum_{i\in [q]} \sfV(x; r,\pi, \sigma_i)$
            \item Output $(r,\est)$
        \end{enumerate}
    \end{framed}
    Note that $\rho$ is the randomness of $\Sim$ and $\sigma_1, \dots, \sigma_q$ serve as the randomness of $\sfV$.
    Assume that $\cA$ is a PPT algorithm for inverting $f_x$'s. Let $\tau = 1/p(n)$. 
    Construct an efficient algorithm $\cB$ with oracle access to $\cA$ as follows:
    \begin{framed}
        \noindent \underline{$\cB^\cA(x; r)$}:
        \begin{enumerate}
            \item For $a = q, q-1, \dots, 1$ in decreasing order
            \begin{itemize}
                \item $(\rho; \sigma_1,\dots, \sigma_q) \gets \cA(r, a/q)$
                \item If $f_x(\rho; \sigma_1,\dots, \sigma_q) = (r, a/q)$
                \begin{itemize}
                    \item ${\pi} \gets \Sim_2(x; \rho)$
                    \item $\hat{\sigma}_1, \dots, \hat{\sigma}_q \gets \cU$
                    \item $\hat{\est} \gets \frac{1}{q} \sum_{i\in [q]} \sfV(x; r,\pi, \hat{\sigma}_i)$
                    \item If $\abs{\hat{\est} - a/q}<t$, output $(a/q, \pi)$
                \end{itemize}
            \end{itemize}
            \item Output $(0,\bot)$ otherwise
        \end{enumerate}
    \end{framed}
    Let $\cB_1$ and $\cB_2$ be the algorithm that outputs the first and the second part of $\cB$'s output respectively and let $\tsfP^\cA = \cB_2^\cA$ be an efficient prover's strategy on input $(x; r)$. 
    Then, the reduction $R^\cA$ runs the protocol $\prot{\tsfP^\cA}{\sfV}$ itself and accepts if the protocol accepts. For convenience, denote $\tsfP = \tsfP^\cA$.
    
    We prove that $R$ and $\cF$ satisfy condition \ref{it:nizk_yes}. For $x\in \cL_n$, suppose that $\cA$ distributionally inverts $f_x$ with deviation at most $1/p(n)$, that is 
    \begin{equation*}
    \label{eq:nizk_inv_rd}
        \Delta_s(\cA(f_x(\cU)), f_x(\cU); \cU, f_x(\cU)) \le 1/p(n).
    \end{equation*}

    Suppose that $(\est^*, \pi^*)$ is the value on which $\tsfP(x;r)$ returns, by a Chernoff bound, 
    with an overwhelming probability $(1 - \negl(n))$ we have 
    \begin{equation*}
        \abs{\est^* - \mathbb{E}\left[{\sfV(x; r,\pi^*)}\right]} < 2\tau, 
    \end{equation*}
    which implies 
    \begin{equation*}
        \abs{\Expec{\substack{r\gets \gen(1^n), \pi\gets \tsfP(x; r)\\ \sigma \gets \cU}}{\sfV(x; r, \pi, \sigma) }- \Expec{r \gets \gen(1^n)}{\cB_1(x; r)}} \le \frac{2}{p(n)} + \negl(n). 
    \end{equation*}
    Analogous to the argument shown before, by setting $\cB_1$ and $\cB_{2,1}$ in the proof of Claim \ref{claim:crpc_est_bd} to be $\cB$ and $\sfV$, we have that 
    \begin{align*}
        \Expec{r \gets \gen(1^n)}{\cB_1(x; r)} &\ge \Expec{\substack{r \gets \gen(1^n), \pi\gets \sfP(x, w;r) \\ \sigma \gets \cU}}{\sfV(x;r, \pi, \sigma)} - \ezn - \frac{1}{p(n)} - \negl(n)\\
        &\ge 1 - \ecn - \ezn - \frac{1}{p(n)} - \negl(n).
    \end{align*}
    Therefore,
    \begin{equation*}
        \Prob{}{R^\cA(x) = 1} = \Expec{\substack{r\gets \gen(1^n) \\ \pi\gets \tsfP(x; r)\\ \sigma \gets \cU}}{\sfV(x; r, \pi, \sigma)} \ge 1 - \ecn - \ezn - \frac{3}{p(n)} - \negl(n). 
    \end{equation*}

    Condition \ref{it:nizk_no_rd} holds, since when $x\notin \cL$, for any polynomial-time algorithm $\cA$, the soundness guarantees 
    \begin{equation*}
        \Prob{}{R^\cA(x) = 1} = \Expec{\substack{r\gets \gen(1^n) \\ \pi\gets \tsfP(x; r)\\ \sigma \gets \cU}}{\sfV(x; r, \pi, \sigma)} \le \epsilon_s(n). 
    \end{equation*}
\end{proof-sketch-of}